\documentclass[10pt]{article}

\usepackage{amsfonts,amssymb,amsmath,relsize}
\usepackage{epsfig,rotating,framed,enumerate,listings,nicefrac}
\usepackage{marginnote}
\usepackage{hyperref}
\usepackage{paralist,amsthm}
\usepackage{authblk}

\usepackage[paper=a4paper,left=30mm,right=30mm,top=30mm, bottom=25mm]{geometry}

\newcommand{\IS}{\text{I}}
\newcommand{\IM}{\text{IM}}
\newcommand{\MIS}{\text{MIS}}
\newcommand{\is}{\text{i}}
\newcommand{\ic}{\text{c}}
\newcommand{\cm}{\text{cm}}
\newcommand{\mc}{\text{mc}}
\newcommand{\im}{\text{im}}
\newcommand{\mis}{\text{mis}}
\newcommand{\MC}{\text{MC}}
\newcommand{\CM}{\text{CM}}
\newcommand{\C}{\text{C}}

\newcommand{\bigo}{\ensuremath{\mathcal{O}}}

\newcommand{\p} {\mbox{P}}

\newcommand{\gansfuss}[1]{\mbox{``}{#1}\mbox{''}}
\newenvironment{desctight}
  {\begin{list}{}{
\setlength\labelwidth{-5pt}
        \setlength{\itemsep}{0.5pt}
        \setlength{\parsep}{0pt}
        \setlength\itemindent{-\leftmargin}
        
}}
{\end{list}}

\newcommand{\lfa}{\mathop{\mathlarger{\mathlarger{\mathlarger{\forall}}}}}

\newtheorem{theorem}{Theorem}
\newtheorem{definition}[theorem]{Definition}
\newtheorem{lemma}[theorem]{Lemma}
\newtheorem{example}[theorem]{Example}
\newtheorem{corollary}[theorem]{Corollary}
\newtheorem{observation}[theorem]{Observation}
\newtheorem{proposition}[theorem]{Proposition}
\newtheorem{remark}[theorem]{Remark}

\begin{document}

\title{Counting and Enumerating Independent Sets with Applications to Knapsack Problems\thanks{A short version of this paper 
has been presented at the {\em International Conference on Operations Research} (OR 2017) \cite{GR17a}.}}

\author[1]{Frank Gurski}
\author[1]{Carolin Rehs}

\affil[1]{\small University of  D\"usseldorf,
Institute of Computer Science, Algorithmics for Hard Problems Group,\newline 
40225 D\"usseldorf, Germany.
\{frank.gurski, carolin.rehs\}@hhu.de }

\date{\vspace{-5ex}}

\maketitle

\begin{abstract}
We introduce methods to count and
enumerate all maximal independent, all maximum independent sets, and all independent 
sets in threshold graphs and $k$-threshold graphs. Within threshold graphs and $k$-threshold graphs
 independent sets 
correspond to feasible solutions in related knapsack instances. 
We give several characterizations for knapsack instances and multidimensional knapsack instances
which allow an equivalent graph.
This allows us to solve special knapsack instances as well as special
multidimensional knapsack instances for fixed number of dimensions in polynomial time. 
We also conclude lower bounds on the number of necessary bins  within several bin packing problems.

\end{abstract}

%%%%%%%%%%%%%%%%%%%%%%%%%%%%%%%%%%%%%%%%%%%%%%%%%%%%%%%%%%%%%%%%%%%%%%%%%%
\section{Introduction}
%%%%%%%%%%%%%%%%%%%%%%%%%%%%%%%%%%%%%%%%%%%%%%%%%%%%%%%%%%%%%%%%%%%%%%%%%%

The knapsack problem is one of the most famous NP-hard tasks in combinatorial optimization.
Within the knapsack problem (KP) there is given a set $A$ of $n$ items. 
Every item $j$ has a profit $p_j$ and a size $s_j$. 
Further there is a capacity $c$ of the knapsack. The task is to choose a subset $A'$ of $A$, 
such that the total profit of $A'$ is maximized and the total size of $A'$ is at most $c$.
The d-dimensional knapsack problem (d-KP) generalizes the knapsack problem
by using items of sizes within a number $d$ of dimensions.

Different techniques for solving hard problems were successfully applied
on the knapsack problem. 
Among these are pseudo-polynomial algorithms, approximation algorithms, 
and integer programming, see \cite{Fre04} for a survey. 
These results lead to several para\-me\-terized algorithms \cite{AGRY16}.
In this paper we show how to solve special knapsack instances as well as special
multidimensional knapsack instances for fixed number of dimensions in polynomial time. 
A related approach to solve special instances for d-KP
was suggested by Chv{\'a}tal and Hammer in \cite{CH77}.
They consider d-KP instances  allowing only zero-one sizes given by 
some matrix $A$ such that some graph model $G(A)$ is a threshold graph. 
Every such instance was solved in time $\bigo(d\cdot n^2)$ by using a split 
partition of graph $G(A)$.

This paper is organized as follows.
In Section \ref{sec-pre}, we give preliminaries on graphs
and special vertex sets in graphs. 
In Section \ref{sec-knaps}, we give a graph theoretic approach 
to solve special knapsack instances. For this purpose we apply threshold graphs, 
which have the useful property, 
that they are equivalent to some knapsack instance, i.e. their independent sets correspond to
feasible solutions in some knapsack instance. We give characterizations for knapsack instances
which allow an equivalent threshold graph. Further we give  methods
in order to count and enumerate all maximal independent sets in threshold graphs. 
These approaches improve the result of \cite{OV04}.
This allows us to solve every KP instance on $n$ items 
which has an equivalent graph in time $\bigo(n^2)$. 
Furthermore we show how to count and enumerate
all maximum independent sets and independent sets in threshold graphs. We also give 
lower bounds on the solutions for the bin packing problem using equivalent graphs.
The idea comes from \cite{CLR04} but the authors did not characterize which instances
of the bin packing problem can be treated in this way.
In Section \ref{sec-d-knaps}, we consider $k$-threshold graphs to 
generalize our results for d-KP.
We characterize d-KP instances which 
allow an equivalent $k$-threshold graph. By combining the 
maximal independent sets for $k$ covering threshold graphs we give a method to
count and enumerate all maximal independent sets in a $k$-threshold
graph in time $\bigo(n^{2k+1})$. Our method implies new bounds on the
number of maximal independent sets in $k$-threshold graphs using the clique
number of  $k$ covering threshold graphs.
For $k=2$ we can improve our results
by considering split partitions for the two covering threshold graphs and
obtain all maximal independent sets in a $2$-threshold
graph in time $\bigo(n^{3})$.
This allows us to solve every d-KP instance  on $n$ items in which every dimension 
has an equivalent graph in time $\bigo(n^{2d+1})$. 
This result generalizes the above mentioned method by Chv{\'a}tal and Hammer in \cite{CH77} 
for solving d-KP instances since we consider 
instances using non-negative real valued sizes and capacities as well as
$k$-threshold graphs as graph models. Since every maximum independent set 
is also a maximal independent set our results improve existing 
solutions for the maximum independent set problem on $k$-threshold graphs 
from \cite{CLR04} if we can bound the vertex degree of the threshold graphs.
In the final Section \ref{sec-concl},  we survey our results on
counting and enumerating special independent sets 
within threshold graphs and $k$-threshold graphs (cf.\ Table \ref{summary}).
Further we give some conclusions on counting and enumerating special  cliques in graphs 
of bounded threshold intersection dimension.

%%%%%%%%%%%%%%%%%%%%%%%%%%%%%%%%%%%%%%%%%%%%%%%%%%%%%%%%%%%%%%%%%%%%%%%%%%
\section{Preliminaries}\label{sec-pre}
%%%%%%%%%%%%%%%%%%%%%%%%%%%%%%%%%%%%%%%%%%%%%%%%%%%%%%%%%%%%%%%%%%%%%%%%%%

We work with finite undirected graphs 
$G=(V,E)$, where $V$ is a finite set of {\em vertices} and 
$E \subseteq \{ \{u,v\} \mid u,v \in V_G,~u \not= v\}$ is a finite set of {\em edges}.
%Let $G=(V,E)$ be some graph. 
For $v\in V$ and $S\subseteq V$ we define
the neighbourhood of $v$ in $S$ by $N(v,S)=\{u\in S~|~ \{u,v\}\in E\}$
and the non-neighbourhood of $v$ in $S$ by $\overline{N(v,S)}=S-(N(v,S)\cup\{v\})$.
The {\em (edge) complement graph} $\overline{G}$ of graph $G$ has the same 
vertex set as $G$ and two vertices in $\overline{G}$ are adjacent if and 
only if they are not adjacent in $G$, i.e. 
$\overline{G}=(V,\{\{u,v\}~|~u,v\in V, u\neq v, \{u,v\}\not\in E\})$.

%A subset $V'\subseteq V$ such that no two vertices of $V'$ are adjacent 
%is called an {\em independent set} of $G$. 
An {\em independent set} of $G$ is a subset $V'$ of $V$ such that
there is no edge in $G$ between two vertices from $V'$.
A {\em maximum independent set} is an independent set of largest size.
A {\em maximal independent set} is an independent
set that is not a proper subset of any other independent set.
Note that a maximum independent set is maximal but in general the converse 
is not true. While computing a maximum independent set is a hard problem,
a maximal independent set can always be found in polynomial time.
The family of all independent sets 
(maximum independent sets and maximal independent sets, respectively) 
of some graph $G$ is denoted by $\IS(G)$ ($\IM(G)$ and $\MIS(G)$, respectively).
The cardinalities of these families are denoted by $\is(G)$ ($\im(G)$ and $\mis(G)$, respectively).
These notations imply the following relations for every graph $G$ it holds
\begin{equation}
\IM(G) \subseteq \MIS(G) \subseteq  \IS(G) \text{~~ and ~~}
\im(G) \leq \mis(G) \leq  \is(G) \label{eq-sizes}
\end{equation}

Enumerating and counting maximal independent sets in
graphs are often studied problems in the field of
special graph classes.
In general every graph on $n$ vertices has at most $3^{n/3}$ maximal independent sets \cite{MM65} 
and each of these sets can be generated in time $\bigo(n^2)$ \cite{PU59}. 
Thus all maximal independent sets can be found in time $\bigo^*(3^{n/3})=\bigo^*(1.4423^n)$.
This number can be achieved as the example of $\frac{n}{3}$ triangles shows.
For $k$-threshold graphs we will show how to compute efficiently better bounds 
on the number of maximal independent sets.

For several natural counting problems it is known that they are 
$\#\p$-complete\footnote{The class of $\#\p$-complete problems is a class of computationally
equivalent counting problems which 
are at least as difficult as NP-complete problems.} \cite{Val79}.
Even for special graph classes counting the number
of maximal independent sets is $\#\p$-complete, such as 
%for bipartite graphs (follows by a result given in \cite{PB83}),
%kann raus, planar bip ist besser:
for regular graphs and 
planar bipartite graphs of maximum degree 5 \cite{Vad01}, and for
chordal graphs \cite{OUU08}. For subclasses of chordal graphs, namely 
threshold graphs and split graphs we will show linear time solutions 
for counting the number of maximal independent sets.

Enumeration problems are often easier than counting problems 
when considering the running time w.r.t. 
the size of the output.
There are special graphs for which enumerating all maximal independent sets can be done
in polynomial time or even constant time w.r.t. the size and the number of independent
sets of the input graph such as  interval, circular arc, and chordal graphs \cite{Leu84}.

The family of all cliques
(maximum cliques and maximal cliques, respectively) 
of some graph $G$ is denoted by $\C(G)$ ($\CM(G)$ and $\MC(G)$, respectively).
The cardinality of these families are denoted by $\ic(G)$ ($\cm(G)$ and $\mc(G)$, respectively).
For every graph $G$  it holds
\begin{equation}
\CM(G) \subseteq \MC(G) \subseteq  \C(G) \text{~~ and ~~}
\cm(G) \leq \mc(G) \leq  \ic(G) \label{eq-sizes-cl}
\end{equation}

As usual let $\alpha(G)$ be the size of a largest independent set and
$\omega(G)$ be the size of a largest clique.
%, $\chi(G)$ the chromatic number, $\theta(G)$
%be the clique cover number of $G$. 
%Further let $m(G)$ be the  number of maximal cliques
%and $\ell(G)$  be the  number of maximal independent sets 
%in $G$.

%%%%%%%%%%%%%%%%%%%%%%%%%%%%%%%%%%%%%%%%%%%%%%%%%%%%%%%%%%%%%%%%%%%%%%%%%%
\section{Knapsack Problem and Threshold Graphs}\label{sec-knaps}
%%%%%%%%%%%%%%%%%%%%%%%%%%%%%%%%%%%%%%%%%%%%%%%%%%%%%%%%%%%%%%%%%%%%%%%%%%

\subsection{Knapsack Problem}

We recall the definition of {\sc Max Knapsack}, see \cite{KPP10} 
for a survey.

\begin{desctight}
\item[Name] {\sc Max Knapsack} ({\sc Max KP})

\item[Instance] A set $A=\{a_1,\ldots,a_n\}$ of $n$ items, 
for every item $a_j$, there 
is a size of $s_j$ and a profit of $p_j$. 
Further there is a capacity $c$ 
for the knapsack.

\item[Task] Find a subset $A'\subseteq A$ such that 
the total profit of $A'$  is maximized %($\sum_{a\in A'}$)
and 
the total size of $A'$ is at most $c$.
\end{desctight}

In this paper, the parameter $n$ is assumed to be a positive integer
and parameters $p_j$, $s_j$, and $c$ are assumed to be 
non-negative reals.
Let $I$ be an instance for {\sc Max KP}.
Every subset $A'$ of $A$ such that $\sum_{a_j\in A'}s_j\leq c$
is a {\em feasible solution} of $I$. A feasible solution
which is not the subset of another feasible solution is called {\em maximal}. 

\begin{definition}\label{def-eq-k}
An instance $I$ for {\sc Max KP} and a graph $G=(V,E)$
are {\em equivalent}, if there is a bijection $f:A\to V$ such that $A'\subseteq A$
is a feasible solution of $I$ if and only if $f(A'):=\{f(a_j)~|~a_j\in A'\}$ is an independent 
set of $G$.
\end{definition}

Next we give some examples of graphs which possess an equivalent 
instance for {\sc Max KP}.
As usual we denote by $P_n$ the path on $n$ vertices,
by $C_n$ the cycle on $n$ vertices, by  $K_n$ the complete graph on $n$ 
vertices, and by $K_{n,m}$ the complete bipartite graph on $n+m$ vertices.

\begin{example}\label{ex1}
\begin{enumerate}
\item\label{ex1-a}
Every clique $K_n$ on $n$ vertices is 
equivalent to  every\footnote{Since the notation of equivalence between
graphs an instances of (cf.\ Definition \ref{def-eq-k}) does not consider the
profits, there are several instances with a given list of sizes and a capacity.} 
instance for {\sc Max KP} with
sizes $s_j=1$ for $1\leq j \leq n$ and capacity $c=1$. 

\item\label{ex1-b}
Every edgeless graph $I_n$ on $n$ vertices  is
equivalent to every instance for {\sc Max KP} with
sizes $s_j=1$ for $1\leq j \leq n$ and capacity $c=n$.

\item\label{ex1-c}
Every star $K_{1,p}$ is 
equivalent to every instance for {\sc Max KP} with
sizes $s_j=1$ for the $p$ items corresponding to the $p$ vertices of degree one, 
$s_j=p$ for the item corresponding to the  vertex of degree $p$, and 
capacity $c=p$. 
\end{enumerate}
\end{example}

Further examples of graphs which possess an equivalent 
instance for {\sc Max KP} are shown in Table \ref{gr}.

\begin{table}[ht]
\begin{center}
\begin{tabular}{cccccccccccccc}  
\\
\epsfig{figure=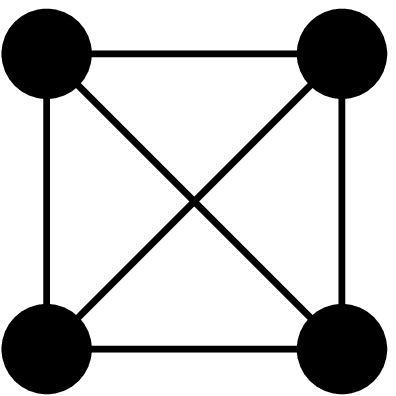,width=1.2cm}&~~~~~~&   \epsfig{figure=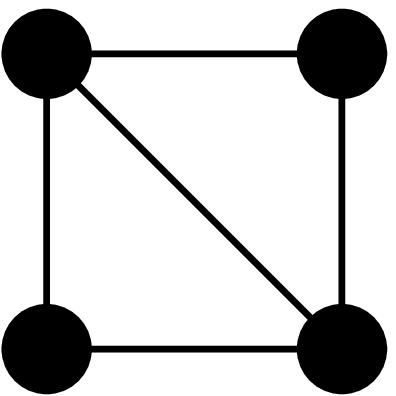,width=1.2cm}&~~~~~~&   \epsfig{figure=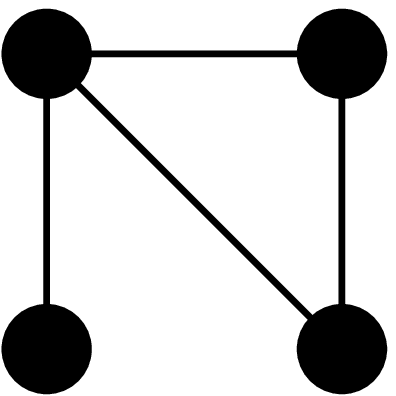,width=1.2cm}&~~~~~~&   \epsfig{figure=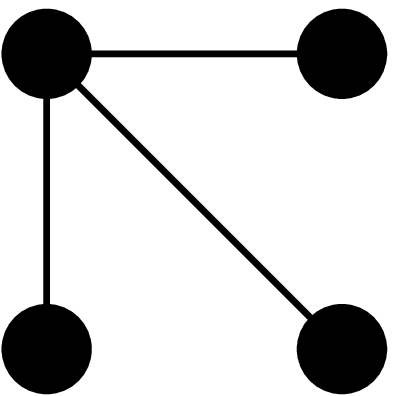,width=1.2cm}&~~~~~~&   \epsfig{figure=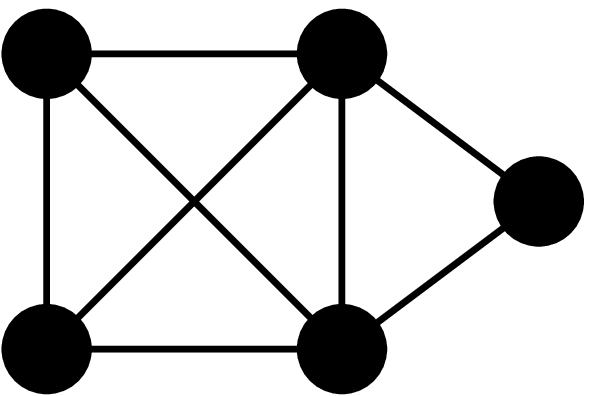,width=1.8cm}\\
$K_4$ && diamond&& paw && claw && cross-house\\
\end{tabular}
\end{center}
\caption
{Special threshold graphs}
\label{gr}
\end{table}

Next we want to characterize graphs for which there is an equivalent
instance for {\sc Max KP}.

\subsection{Threshold Graphs}\label{sec-thresh}

We will use the following characterizations
for threshold graphs, see \cite{MP95a} for a survey.

\begin{theorem}[Theorem 1.2.4 of \cite{MP95a}]\label{th-thres}
For every graph $G=(V,E)$ the following statements are equivalent.
\begin{enumerate}

\item\label{c1}  There exist non-negative reals $w_v$, $v\in V$, and $T$ 
such that for every $U\subseteq V$ it holds  $\sum_{v\in U}w_v\leq T$ if and
only if $U$ is an independent set of $G$.

\item Graph $G$ contains no $C_4$, no $P_4$, and no $2K_2$ as
induced subgraph.

\item\label{c-edge}  There exist non-negative reals $w_v$, $v\in V$, and $T$ 
such that for every two vertices $u$, $v$ it holds  $w_u + w_v > T$ if
and only if  $\{u,v\}\in E$.

\item\label{c-seq} Graph  $G$ can be constructed from the one-vertex graph $K_1$ by repeatedly adding
an isolated vertex or a dominating vertex.

\item Graph $G$ is a split graph with vertex partition $V=S\cup K$ 
and the neighbourhoods of the vertices of the
independent set $S$ are nested.

%\item Graph $G$ and its edge complement graph $\overline{G}$ are trivially perfect.
\end{enumerate}
\end{theorem}

Condition (\ref{c1}) of Theorem \ref{th-thres} was used 
by Chv\'atal and Hammer in the 1970s to define threshold graphs \cite{CH73,CH77}.
By \cite[page 221]{Gol80} one can assume that all $w_{v}$ and $T$ are non-negative
integers. On the other hand, the possibility of choosing arbitrary (also negative)
reals leads a larger graph class, namely {\em generalized threshold graphs}, see \cite{RRS89}.
Furthermore condition (\ref{c1}) of Theorem \ref{th-thres} implies 
a characterization for graphs which are equivalent to instances 
for {\sc Max KP}.\footnote{Robinson \cite{Rob97} has introduced
{\em knapsack graphs} as graphs whose independent sets correspond to feasible solutions
of a knapsack instance. In \cite{Rob97} knapsack graphs are characterized 
as graphs that do not contain 
a cycle of four vertices $C_4$, 
a path on four vertices $P_4$, or 
two disjoint edges $2K_2$ 
as an induced subgraph. Further it is shown that every
knapsack graph can be constructed from the one-vertex graph $K_1$ by
repeatedly adding an isolated or a dominating vertex.
Further it was shown that
the edge complement of a knapsack graphs is also a knapsack graph. 
These results have already been shown in the 1970s by 
Chv\'atal and Hammer in \cite{CH73,CH77} which implies that the class
of knapsack graphs equals the well known class of threshold graphs.}

\begin{observation}\label{obs-kn-th}
A graph is a threshold graph if and only if it has an equivalent
instance for {\sc Max KP}.
\end{observation}

Threshold graphs can be recognized in linear time \cite{CH73,HIS78}. Recently
a linear time recognition algorithm was found which also gives a forbidden
induced subgraph from $\{2K_2,P_4,C_4\}$ if the input is not a threshold graph \cite{HK07}.
For a set of graphs ${\mathcal F}$ we denote by {\em ${\mathcal F}$-free graphs} 
the set of all graphs that do not contain a graph of ${\mathcal F}$ as an induced subgraph.

\begin{proposition}[\cite{BLS99}]\label{prop-classes} 
We have the following properties for threshold graphs.
\begin{enumerate}
\item threshold $\subset$  trivially perfect $\subset$ co-graphs
%\item threshold $\subset$ interval  $\subset$ chordal
\item threshold $\subset$ split  $\subset$ chordal
\item threshold $\subset$ split  $\subset$ $\{2K_2,C_4\}$-free graphs
\end{enumerate}
\end{proposition}

By condition (\ref{c-seq}) of Theorem \ref{th-thres} every threshold
graph $G=(V,E)$ with $V=\{v_1,\ldots,v_n\}$ can be constructed from the one-vertex graph  by repeatedly adding
an isolated vertex or a dominating vertex. 
A {\em creation sequence} for $G$ (cf.~\cite{HSS06}) is 
a binary string $t=t_1,\ldots,t_n$ of length $n$  such that
there is a bijection $v:\{1,\ldots,n\}\to V$ with 
\begin{itemize}
\item
$t_i=1$ if  $v(i)$ is a
dominating vertex for the graph induced by $\{v(1),\ldots,v(i)\}$ and 
\item
$t_i=0$ if  $v(i)$ is an isolated vertex for the graph induced by $\{v(1),\ldots,v(i)\}$.
\end{itemize}
W.l.o.g. we define a single vertex to be a dominating vertex, i.e. $t_1=1$.

\begin{example}\label{ex-creat}
In Table \ref{cr-sq} we list all creation sequences for threshold graphs on four vertices
and the defined graphs.
\end{example}

\begin{table}[ht]
\begin{center}
\begin{tabular}{|l|ll||l|l|}
\hline
sequence  & graph             & & sequence & graph \\
\hline
$1000$ & edgeless graph $I_4$ & & $1100$   & co-diamond\\
$1001$ & claw                 & & $1101$   & paw         \\
$1010$ & co-paw               & & $1110$   & co-claw     \\
$1011$ & diamond              & & $1111$   & complete graph $K_4$\\
\hline
\end{tabular}
\end{center}
\caption{Creation sequences for all threshold graphs on four vertices}
\label{cr-sq}
\end{table}

Using the method in Figure 1.4 of \cite{MP95a} a creation sequence
can be found in linear time.

\begin{lemma}[\cite{MP95a}]\label{find-cs}
Given some threshold graph $G$  on $n$ vertices and $m$ edges, a creation sequence for $G$ 
can be found in time  $\bigo(n+m)$.
\end{lemma}

The following properties of creation sequences will be useful in Sections \ref{ce-mis-th}, \ref{th-count-maxis}, and
\ref{th-count-is}.

\begin{observation}\label{obs-creation-se}Let $G$ be a threshold graph and $t=t_1,\ldots,t_n$ 
be a creation sequence for $G$.
\begin{enumerate}
\item
Every vertex corresponding to a $1$ in the sequence is adjacent to all
vertices corresponding to a $1$ or $0$ on the left and to all vertices corresponding to a $1$  
on the right. 
\item
Every vertex corresponding to a $0$ in the sequence is adjacent 
to all vertices corresponding to a $1$  
on the right.

\item  \label{obs-d-i}
For every $j$ such that $t_j=1$ the vertex set $\{v(j)\}\cup\{v(i)~|~ j<i, t_i=0\}$ leads a maximal 
independent set.

\item  \label{obs-c-i}
For every $j$ such that $t_j=0$ the vertex set $\{v(j)\}\cup\{v(i)~|~ j<i, t_i=1\}$ 
and set $\{v(1)\}\cup\{v(i)~|~ j<i, t_i=1\}$  lead a maximal clique.

\item  \label{obs-d}
The vertex
set $\{v(i)~|~i=1 \vee t_i=0\}$ leads a maximum independent set.

\item \label{obs-c}
The vertex set $\{v(i)~|~ t_i=1\}$  leads a maximum clique.

\item \label{obs-compl}
A creation sequence $t'$ for the complement graph $\overline{G}$ can be obtained by $t'_1=1$ and
$t'_i=1-t_i$ for $1< i \leq n$.
\end{enumerate}
\end{observation}

\begin{corollary}\label{find-maxcln}
Let $G$ be a threshold graph on $n$ vertices and $m$ edges.  
The size of a maximum clique $\omega(G)$ and
the size of a maximum independent set $\alpha(G)$ 
can be found in time $\bigo(n+m)$.
\end{corollary}

By their characterization as  $\{2K_2,P_4,C_4\}$-free graphs 
and also by Observation \ref{obs-creation-se} (\ref{obs-compl})
we
conclude that threshold graphs are closed under taking edge complements.

\begin{lemma}\label{comple}
The complement of a threshold graph is a threshold graph.
\end{lemma}

The next result immediately follows by the characterization using creation sequences.

\begin{observation}\label{number}
There are $2^{n-1}$ many threshold graphs on $n$ vertices.
\end{observation}

%%%%%%%%%%%%%%%%%%%%%%%%%%%%%%%%%%%%%%%%%%%%%%%%%%%%%%%%%%%%%%%%%%%%%%
\subsection{Creating knapsack problems from threshold graphs}\label{sec-kg-to-kp}
%%%%%%%%%%%%%%%%%%%%%%%%%%%%%%%%%%%%%%%%%%%%%%%%%%%%%%%%%%%%%%%%%%%%%%

Given some threshold graph $G=(V,E)$ we know by definition, that
there is an equivalent {\sc Max KP} instance for $G$, which can be
identified in several ways.
Using a creation sequence $t_1,\ldots,t_n$, which can be
found in linear time by Lemma \ref{find-cs}
we can define an equivalent instance for {\sc Max KP} as shown 
in Figure \ref{fig:thres_to_knap}.

\begin{figure}[ht]
%\hrule
%{\strut\footnotesize \bf Algorithm {\sc ....}} 
\hrule
\medskip
\begin{tabbing}
xxx \= xxx \= xxx \= xxx \= xxx\= xxx \= xxx \kill
$c=1$;  \\
$s_1=1$; \\
for ($i = 2$; $i\leq n$; $i++$) \\
\>    if ($t_i=0$) \>\>\>\>  $\vartriangleright$  $v(i)$ is an isolated vertex in $G[\{v(1),\ldots,v(i)\}]$ \\
\> \> $s_i=1$;  \\
\> \> $c=2c+1$; \\
\> \> for ($j = 1$; $j\leq i-1$; $j++$) \\
\>\>\> $s_j=2\cdot s_j$ \\
\>   else   \>\>\>\>  $\vartriangleright$ $v(i)$ is a dominating vertex in $G[\{v(1),\ldots,v(i)\}]$  \\
\> \> $s_i=c$; 
\end{tabbing}
\hrule
\caption{Creating an equivalent instance for {\sc Max KP} from a threshold graph given by
a creation sequence $t_1,\ldots,t_n$}
\label{fig:thres_to_knap}
\end{figure}

\begin{lemma}\label{le-gr-k}
For some given threshold graph on $n$ vertices an equivalent 
instance for {\sc Max KP} can be
constructed in time $\bigo(n^2)$.
\end{lemma}

\begin{proof}
Let $G$ some threshold graph on $n$ vertices and $m$ edges. By Lemma \ref{find-cs} 
we can find a creation sequence for $G$ in time $\bigo(n+m)$
and by the algorithm given in  Figure \ref{fig:thres_to_knap} 
we obtain an instance $I$ for {\sc Max KP} in time $\bigo(n^2)$.
\end{proof}

Obviously we can start with an arbitrary value for $c$ and $s_1$ in Figure \ref{fig:thres_to_knap} such that
there are several equivalent instances for {\sc Max KP} for some fixed threshold
graph.
In general the sizes found by the algorithm given in  Figure \ref{fig:thres_to_knap}
are not minimal, see Example \ref{ex1}.\ref{ex1-b} and \ref{ex1}.\ref{ex1-c}. 
By combining consecutive isolated and dominating
vertices one can produce smaller sizes in many cases. In \cite{Orl77} (see also
Section 1.3  of \cite{MP95a}) it
is shown how to find minimal sizes.

%%%%%%%%%%%%%%%%%%%%%%%%%%%%%%%%%%%%%%%%%%%%%%%%%%%%%%%%%%%%%%%%%%%%%%
\subsection{Creating threshold graphs from knapsack problems}\label{sec-graph-fr-pr}
%%%%%%%%%%%%%%%%%%%%%%%%%%%%%%%%%%%%%%%%%%%%%%%%%%%%%%%%%%%%%%%%%%%%%%

Defining a graph from an instance for {\sc Max KP}
is more difficult, since not every {\sc Max KP}
instance allows us to give an equivalent graph
by the following example.

\begin{example}\label{ex-no-kp-i}
Let $I$ be defined by a set $A$ of 5 items with sizes 
$s_1=12$, $s_2=10$, $s_3=11$, $s_4=8$, $s_5=9$ and capacity $c=26$. Then the set
of feasible solutions includes of all  subsets of size two of $A$.
If there is an equivalent graph for $I$ 
it has no edge between two
vertices. But since every subset of three items does not lead a
feasible solution, this instance does not allow an equivalent
graph.
\end{example}

The situation  of the given example can be generalized as follows. 
Let $I$ be an instance for {\sc Max KP}
on item set $A$. If $A$ has a
subset $A'$ of at least three items such that for every two different
items $a_i$ and $a_j$ of $A'$ it holds $s_i+s_j\leq c$ 
and $\sum_{a_i\in A'}s_i>c$ then the corresponding
instance does not allow an equivalent graph. By avoiding this
situation we next will  characterize {\sc Max KP} instances allowing an equivalent
graph. Therefor let $ \mathcal{I}$ be the set of all instances  of {\sc Max KP}. 
We define a Boolean property $P: \mathcal{I} \to \{\text{true},\text{false}\}$ 
for some instance $I\in \mathcal{I}$ by
\begin{equation}
P(I)=\lfa_{A'\subseteq A}\big((\lfa_{a_j,a_{j'}\in A'} s_j+s_{j'}\leq c) \Rightarrow \sum_{a_j\in A'}s_j\leq c\big).
\end{equation}

The idea of this property is to ensure feasibility in the case of independence
within $A$. Although independence has to be valid within a graph, this will be useful.
For subsets $A'$ with $|A'|\leq 2$ the property is always true.
Since we assume that $j\neq j'$ the implication 
from the right to the left is always true.\footnote{For $A=A'=\{a_1,a_2\}$, $s_1=2$, $s_2=7$, $c=10$,
the right condition is true, but for $j=j'=2$ the left condition is not true.}

\begin{lemma}\label{lemma-g-p}
If some instance $I$ for {\sc Max KP} has an equivalent
graph, then $I$ satisfies property $P(I)$.
\end{lemma}

\begin{proof}
Let $I$ be some instance for {\sc Max KP}. In order to give a proof by
contradiction we assume that  $I$ does not satisfy property $P(I)$.
That is, there is some $A'\subseteq A$ such that $\forall_{a_j,a_{j'}\in A'} s_j+s_{j'}\leq c$
and $\sum_{a_j\in A'}s_j> c$. Then there is some set $A''\subseteq A'$ and $a_\ell\in A''$
such that $\sum_{a_j\in A''}s_j> c$ and $\sum_{a_j\in A''-\{a_\ell\}}s_j\leq c$. That 
means that 
for every $a_j\in A''$ set $\{a_j,a_\ell\}$ is a feasible solution for $I$ 
and $A''-\{a_\ell\}$ is also a feasible solution for $I$, but 
$A''$ is not a feasible solution for $I$. 

Every bijection $f:A\to V$ certifying an equivalent graph $G=(V,E)$ for $I$ has to map 
for every  $a_j\in A''$ item set $\{a_j,a_\ell\}$ onto an independent set $\{f(a_j),f(a_\ell)\}$
and item  set $A''-\{a_\ell\}$  onto an independent set $f(A''-\{a_\ell\})=f(A'')-f(\{a_\ell\})$.
But then also $f(A'')$ is an independent set in $G$ while $A''$ 
is not a feasible solution for $I$.
\end{proof}

Next we want to show the reverse direction of Lemma \ref{lemma-g-p}. For some 
instance $I$ of {\sc Max KP} we define the graph $G(I)=(V(I),E(I))$ by
%$$
%\begin{array}{lcl}
%V(I)&=&\{v_j~|~a_j\in A\} \text{ and} \\
%E(I)&=&\{\{v_{j},v_{j'}\} ~|~ s_{j}+s_{j'}> c\}.
%\end{array}
%$$
\begin{equation}
V(I)=\{v_j~|~a_j\in A\} ~~~ \text{ and } ~~~  E(I)=\{\{v_{j},v_{j'}\} ~|~ s_{j}+s_{j'}> c\}.
\end{equation}

In general $G(I)$ is not equivalent to $I$, see Example \ref{ex-no-kp-i}.
But by Condition (\ref{c-edge}) of Theorem \ref{th-thres} and by our assumptions
for instances of {\sc Max KP} graph $G(I)$
is even a threshold graph.

\begin{observation}\label{obs-g-i}
For every instance $I$ for {\sc Max KP} graph $G(I)$ is a threshold graph.
\end{observation}

Next we give a tight connection between the property $P(I)$ and
graph $G(I)$. Since both are defined by the sizes of $I$, the
result follows straightforward.

\begin{lemma}\label{lemma-p-g}
Let $I$ be 
some instance $I$ for {\sc Max KP}. Then $I$  satisfies property $P(I)$, if and only
if $I$ is equivalent to graph $G(I)$.
\end{lemma}

\begin{proof}
Let $I$ be an instance for {\sc Max KP}. 

First assume that $I$ satisfies property $P(I)$.
Let $A'\subseteq A$ be a feasible solution for $I$, i.e. $\sum_{a_j\in A'}s_j\leq c$.
Then obviously also $\forall_{a_j,a_{j'}\in A'} s_j+s_{j'}\leq c$ holds 
and by the definition of $G(I)$ the vertex set 
$V'=\{v_j ~|~ a_j\in A'\}$ is an independent set in graph $G(I)$.
For the reverse direction let
$V'\subseteq V$ be an independent set in $G(I)$. 
Then for $A'=\{a_j~|~v_j\in V'\}$ 
it holds $\forall_{a_j,a_{j'}\in A'} s_j+s_{j'}\leq c$ and since
property $P(I)$ is satisfied this implies that $\sum_{a_j\in A'}s_j\leq c$, i.e.
$A'$ is a feasible solution for $I$. That is, the bijection defined 
by the definition of $V(I)=\{v_j~|~a_j\in A\}$ verifies that $I$ is equivalent to $G(I)$.    

Next assume that $I$ is equivalent to graph $G(I)$. Let $f$ be 
the bijection between the item set $A$ and the vertex set $V(I)$ which exists by Definition \ref{def-eq-k}. 
Let $A'\subseteq A$. If $\forall_{a_j,a_{j'}\in A'} s_j+s_{j'}\leq c$, then by the definition
of $E(I)$  for all $a_j,a_{j'}\in A'$ it holds $\{f(a_j),f(a_{j'})\}\not\in E$ and thus 
set $V'=\{f(a_j) ~|~ a_j\in A'\}$ is an independent set of $G(I)$, which
corresponds to a feasible  solution $A'$ 
by the equivalence of $I$ and $G(I)$ via bijection $f$. Thus
$\sum_{a_j\in A'}s_j\leq c$, which  implies that $P(I)$ holds true.
\end{proof}

Now we can state a result corresponding to the reverse direction of Lemma \ref{le-gr-k}.

\begin{lemma}\label{le-k-gr}
Let $I$ be some instance for {\sc Max KP} on $n$ items which has an equivalent graph. Then
$I$ is equivalent to graph $G(I)$, which  can be constructed from $I$ in time $\bigo(n^2)$.
\end{lemma}

\begin{proof}
Let $I$ be some instance for {\sc Max KP} on $n$ items which has an equivalent graph.
By Lemma \ref{lemma-g-p} we know that $I$ satisfies property $P(I)$ and by Lemma \ref{lemma-p-g}
we know that
$I$ is equivalent to graph $G(I)$. Furthermore we
can construct $G(I)$  in time $\bigo(n^2)$
by its definition.
\end{proof}

We obtain the following characterizations
for instances $I$ of {\sc Max KP} allowing an equivalent graph.

\begin{theorem}\label{corollary-char}
For every  instance $I$ of {\sc Max KP}  the
following conditions are equivalent.
\begin{enumerate}
\item Instance $I$ has an equivalent graph.
\item Instance $I$ satisfies property $P(I)$.
\item Instance $I$ is equivalent to graph $G(I)$.
\item Instance $I$ has an equivalent threshold graph.
\end{enumerate}
\end{theorem}

\begin{proof}
(1)$\Rightarrow$(2) by Lemma \ref{lemma-g-p}, (2)$\Leftrightarrow$(3)  by Lemma \ref{lemma-p-g}, 
(3)$\Rightarrow$(4)  by Observation \ref{obs-g-i}, (4)$\Rightarrow$(1) obvious
\end{proof}

\begin{corollary}\label{corollary-char-gi}
Let $I$ be some instance for {\sc Max KP}  which 
has two equivalent graphs $G_1$ and $G_2$.
Then $G_1$ is isomorphic to $G_2$.
\end{corollary}

\begin{proof}
Let $G_i=(V_i,E_i)$ for $i=1,2$ be an equivalent graph for $I$ and $f_i:A\to V_i$ the
thereby existing bijections. 
Then by
$$
\{v_j,v_{j'}\} \in E_1 \Leftrightarrow \{f^{-1}_1(v_{j}),f^{-1}_1(v_{j'})\} \text{ is no feasible solution for }I   \Leftrightarrow   \{f_2(f^{-1}_1(v_{j})),f_2(f^{-1}_1(v_{j'}))\}\in E_2 
$$
the isomorphy of $G_1$ and $G_2$ follows.
\end{proof}

Whenever some instance $I$ for {\sc Max KP} has an equivalent graph, 
then by Theorem \ref{corollary-char} graph
$G(I)$ is also an equivalent graph for $I$ thus
we obtain the next result.

\begin{corollary}\label{corollary-char-gi2}
Let $I$ be some instance for {\sc Max KP}  which has an equivalent graph $G$.
Then $G$ is isomorphic to $G(I)$.
\end{corollary}

Although not every instance for {\sc Max KP} allows an equivalent graph there are
several such instances.
Since there are $2^{n-1}$ threshold graphs on $n$ vertices (cf.\ Observation \ref{number})
the given algorithm given in  Figure \ref{fig:thres_to_knap} leads 
$2^{n-1}$ equivalent instances  for {\sc Max KP} on $n$ items. Further every 
of these instances leads further instances allowing an equivalent graph  by 
multiplying every size and the capacity by some common real number. Furthermore the profits
can be chosen arbitrary in such instances.

%http://people.sc.fsu.edu/~jburkardt/datasets/knapsack_01/knapsack_01.html

%%%%%%%%%%%%%%%%%%%%%%%%%%%%%%%%%%%%%%%%%%%%%%%%%%%%%%%%%%%%%%%%%%%%%%
\subsection{Counting and enumerating maximal independent sets}\label{ce-mis-th}
%%%%%%%%%%%%%%%%%%%%%%%%%%%%%%%%%%%%%%%%%%%%%%%%%%%%%%%%%%%%%%%%%%%%%%

The maximal independent sets in threshold graphs
represent maximal feasible solutions, i.e. non-extensible subsets $A'\subseteq A$
such that $\sum_{a_j\in A'}s_j\leq c$, of corresponding instances for {\sc Max KP}.
Every optimal solution is among one of these sets. Therefore we 
want to count and enumerate these sets.

%%%%%%%%%%%%%%%%%%%%%%%%%%%%%%%%%%%%%%%%%%%%%%%%%%%%%%%%%%%%%%%%%%%%%%
\subsubsection{Split Graphs}
%%%%%%%%%%%%%%%%%%%%%%%%%%%%%%%%%%%%%%%%%%%%%%%%%%%%%%%%%%%%%%%%%%%%%%

In oder to list all maximal independent sets of a threshold graph, in \cite{OV04}
a method using their relation to split graphs (cf.\ Proposition \ref{prop-classes}) 
and the hereby existing 
split partition was mentioned without any bound 
on the running time. Since the structure of split graphs will also be useful for our
results in Section \ref{sec-c-e-k-t} we want to show how to list all maximal 
independent sets of a split graph.

A {\em split graph} is a graph $G$ whose vertex  set $V$ can be partitioned 
into a clique $K$ and an independent set $S$ (either of which may be empty).
For the edges between vertices of
$K$ and vertices of $S$ there is no restriction.

For some split graph with a partition of its vertex set in a clique $K$ and an independent set $S$ 
we denote the pair $(K,S)$ as a {\em split partition}. 
In general a spilt partition $(K,S)$ is
not unique and $S$ is is not a maximum independent set and $K$ is not
a maximum clique in $G$. But there is at most one vertex which missing to obtain
a maximum independent set or clique, which is even maximal, respectively.

\begin{theorem}[\cite{HS81}]\label{th-hs} Let $G$ be a split graph 
with split partition $(K,S)$. Then exactly one of the three conditions
hold true.
\begin{enumerate}
\item $|S|=\alpha(G)$ and $|K|=\omega(G)$. In this case the partition is unique.
\item $|S|=\alpha(G)$ and $|K|=\omega(G)-1$. In this case there is one vertex $x\in S$, such that
$K\cup \{x\}$ is a clique.
\item $|S|=\alpha(G)-1$ and $|K|=\omega(G)$. In this case there is one vertex $x\in K$, such that
$S\cup \{x\}$ is an independent set.
\end{enumerate}
\end{theorem}

The special structure of split graphs allows us to count and enumerate
all maximal independent sets as follows. The result can also be applied
to threshold graphs and even gives ideas for the problem on $k$-threshold
graphs (cf.\ Theorem \ref{t-enum-2}).

\begin{theorem}\label{l-is-spli} 
For every split graph $G$ 
the number of maximal independent sets is $\omega(G)$ or  $\omega(G)+1$. 
For every split graph $G$ on $n$ vertices which is given by a split partition
all maximal independent sets 
can be counted and enumerated in time $\bigo(\omega(G)\cdot n)$.
\end{theorem}

\begin{proof}
We determine the family of all maximal independent sets of
of a  
%connected\footnote{Since split graphs are $2K_2$ free, every non-connected split
%graph has at most one non-trivial component
%which contains edges. This allows us to compute the family of all maximal independent sets
%for this component and add the 
%remaining isolated vertices to all of these sets.} 
split graph $G$ with split partition $(K,S)$. 
In order to count the maximal independent sets
we assume that  $K$ is a maximum clique, i.e. $|K|=\omega(G)$, which
easily can be achieved by Theorem \ref{th-hs}.
% %is possible by Lemma \ref{l-th-spl}.
We define $K'$ as the subset of the vertices in $K$ which 
are not adjacent to any vertex in $S$. 
Then the maximal independent sets of $G$ can be obtained as follows.
\begin{itemize}
\item If $K'=\emptyset$, set $S$ is a maximal independent set. 

\item If $K'\neq \emptyset$, every vertex $v\in K'$ leads a 
maximal independent set $S\cup\{v\}$. 

\item Every proper subset $S'$ of $S$ is an independent set but not
a maximal independent set. To go through all these proper subsets 
is very inefficient.
But we know that these sets can be extended by 
exactly one vertex from $K$. Thus every vertex $v\in K-K'$ 
together with its non-neighbours
in $S$ leads a maximal independent set, i.e. every vertex $v\in K-K'$ leads the 
maximal independent set $\overline{N(v,S)}\cup\{v\}$.                              
\end{itemize} 
The union of all sets obtained in both steps leads the family
of all maximal independent sets of $G$.
Since we assume that we have given 
a split partition $(K,S)$ for $G$ the running time
for every of the three steps is $\bigo(n)$. 

In order to justify
the number of maximal independent sets
we consider the three steps given above.
The first step leads at most one set
and the second and third step lead $|K|=\omega(G)$ sets together.
\end{proof}

Please note that
both numbers of maximal independent sets given in Theorem \ref{l-is-spli} are possible. 
The $P_4$ has $3=\omega(P_4)+1$ maximal independent sets and the $K_4$ has $4=\omega(K_4)$ 
maximal independent sets.

By \cite{HS81} split graphs can be recognized and a split partition can be found in linear time.
In \cite{HK07} a $\bigo(n+m)$ time recognition algorithm was found which also gives a forbidden
induced subgraph from $\{2K_2,C_5,C_4\}$ if the input is not a split graph.

\begin{corollary}\label{l-is-spli-3} 
For every split graph $G$ on $n$ vertices and $m$ edges all maximal independent sets 
can be counted and enumerated in time 
$\bigo(\omega(G)\cdot n+m)$.
After $\bigo(n+m)$ time precomputation all independent sets in $G$  can be enumerated 
in constant time per output.
\end{corollary}

%%%%%%%%%%%%%%%%%%%%%%%%%%%%%%%%%%%%%%%%%%%%%%%%%%%%%%%%%%%%%%%%%%%%%%
\subsubsection{Threshold Graphs}\label{sec-mis-th}
%%%%%%%%%%%%%%%%%%%%%%%%%%%%%%%%%%%%%%%%%%%%%%%%%%%%%%%%%%%%%%%%%%%%%%

Since threshold graphs are split graphs (cf.\ Proposition
\ref{prop-classes}), we can use the method of Theorem \ref{l-is-spli} 
to  enumerate all maximal independent sets in a 
threshold graph. From a given creation sequence a split partition can found as follows.

\begin{remark}\label{tresh-spli-p}
For every threshold graph $G$ on $n$ vertices  a split partition $(K,S)$ can be found in
time $\bigo(n)$ from a creation
sequence $t=t_1\ldots t_n$ for $G$. Vertex $v(1)$ can be chosen into $K$ or $S$. For 
$i>1$ if $t_i=0$  vertex $v(i)$ will be chosen into $S$ and if $t_i=1$ vertex
$v(i)$ will be chosen into $K$. Furthermore, if we choose $v(1)$ into $K$ this
set leads a maximum clique, i.e. $|K|=\omega(G)$ and  $|S|=\alpha(G)-1$. 
If we choose $v(1)$ into $S$ this
set leads a maximum independent set, i.e. $|S|=\alpha(G)$ and $|K|=\omega(G)-1$. 
\end{remark}

Next we give a more simple method to enumerate all maximal independent sets in a 
threshold graph.

\begin{figure}[ht]
\hrule
\medskip
\begin{tabbing}
xxx \= xxx \= xxx \= xxx \= xxx\= xxx \= xxx \kill
$A=\emptyset$; \\ 
$\MIS(G)=\emptyset$; \\
for ($i = n$; $i\geq 1$; $i--$) \\
\>    if ($t_i=1$) \\
\> \> $\MIS(G)= \MIS(G) \cup \{A\cup\{v(i)\}\}$;  \\
\>   else   \>\>\>\>  $\vartriangleright$ $t_i=0$\\
\> \> $A=A\cup\{v(i)\}$; 
\end{tabbing}
\hrule
\caption{Enumerating all maximal independent sets in a threshold graph.}
\label{fig:mis_thres}
\end{figure}

\begin{theorem}\label{t-enum2}
The number of all  maximal independent sets in a threshold $G$ equals $\omega(G)$.
Let $G$ be a threshold graph $G$ on $n$ vertices 
which is given by a creation sequence. Then all  maximal independent sets in $G$
can be counted in time $\bigo(n)$ and 
enumerated  in time $\bigo(\omega(G)\cdot n)$.
\end{theorem}

\begin{proof}
Let $G$ be a threshold graph and $t=t_1\ldots t_n$ be a creating
sequence for $G$, i.e. $t_1=1$ and $t_i\in\{0,1\}$. In Ob\-ser\-va\-tion \ref{obs-creation-se} (\ref{obs-d-i})
we mentioned how we can obtain maximal independent sets 
using a creating sequence, which is realized in the method
given in Figure \ref{fig:mis_thres} in order to 
generate all maximal independent sets
in $G$. 
Since the vertices corresponding to a 1 in the creation sequence
lead a maximum clique (Observation \ref{obs-creation-se} (\ref{obs-c})) 
our method implies every threshold graph $G$ has exactly
$\omega(G)$ maximal independent sets.
\end{proof}

\begin{figure}[ht]
\centerline{\epsfig{figure=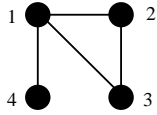,width=2.0cm}}
\caption{Paw graph considered in Example \ref{ex-mis-th}.}
\label{fig:paw-no}
\end{figure}

\begin{example}\label{ex-mis-th} We apply the method
given in Figure \ref{fig:mis_thres} we generate all maximal independent sets
in the paw graph, which can be defined by the creation sequence $t=1101$.
$$
\begin{array}{l|c|c|c|l}
i & v(i) &t_i  & A & \MIS(\text{paw})  \\
\hline
4 &v_1 &1   & \emptyset & \{\{v_1\}\} \\
3 &v_4 &0   &  \{v_4\}  &  \{\{v_1\}\}\\
2 &v_2 &1   &  \{v_4\}  &  \{\{v_1\},\{v_2,v_4\}\}\\
1 &v_3 &1   &  \{v_4\}  &  \{\{v_1\},\{v_2,v_4\},\{v_3,v_4\}\}\\
\end{array}
$$
So we obtain three maximal independent sets which equals
the clique size of the paw graph.
\end{example}

For of graphs $G$ of special graph classes the value of $\omega(G)$
can be bounded, which implies that we only have a small number
of maximal independent sets.  Therefore we recall that
a graph $G$ on $n$ vertices and $m$ edges is {\em $\ell$-sparse} if $m \le \ell \cdot n$. 
It
is {\em uniformly $\ell$-sparse} if every subgraph of $G$ is $\ell$-sparse.

\begin{corollary}\label{th-count-planar-and-so} 
Let $G$ be a threshold graph.
\begin{enumerate}
\item If $G$ is planar, then $\mis(G)\leq 4$. 
\item If $G$ is uniformly $\ell$-sparse, then $\mis(G)\leq 2\ell+1$. 
\item If $G$ has maximum degree at most $d$, then $\mis(G)\leq d+1$. 
\end{enumerate}
\end{corollary}

\begin{proof}\begin{inparaenum}[(1.)]
\item Planar graphs do not contain the $K_5$ as a subgraph, thus
$\omega(G)\leq 4$.
\item If graph $G$ is uniformly $\ell$-sparse then the complete 
graph $K_{2\ell+2}$ is not a subgraph of $G$, thus $\omega(G)\leq 2\ell+1$.
\item If  graph $G$ has maximum degree at most $d$ then the complete 
graph $K_{d+2}$ is not a subgraph of $G$, thus $\omega(G)\leq d+1$.
\end{inparaenum}
\end{proof}

\begin{corollary}\label{cor-enum2}
Let $G$ be a threshold graph on $n$ vertices and $m$ edges.  
%The number of all  maximal independent sets in $G$ equals $\omega(G)$.
All maximal independent sets in $G$ can be counted in time $\bigo(n+m)$.  
All maximal independent sets in $G$  can be enumerated  in time 
$\bigo(\omega(G)\cdot n+m)$. 
After $\bigo(n+m)$ time precomputation all independent sets in $G$  can be enumerated 
in constant time per output.
\end{corollary}

\begin{proof}
Since a creation sequence can be found in time  $\bigo(n+m)$
by Lemma \ref{find-cs} we obtain the result by Theorem \ref{t-enum2}.
Since the time of the computation is at most to the size of the output, 
all independent sets in $G$  can be enumerated  in constant time per output.
\end{proof}

%Please note that in general $\bigo(\omega(G)\cdot n+m)$ is larger than  $\bigo(n+m)$.
%This is the case for an $n$ vertex graph consisting of a clique of size $\log(n)$
%and $n-\log(n)$ isolated vertices. 
%Then  $\bigo(n\log(n))=\bigo(\omega(G)\cdot n+m)\subsetneq \bigo(n+m)=\bigo(n)$.

In \cite{OUU08} results on counting and enumeration problems for chordal graphs, 
especially concerning independent sets, are given. Since threshold graphs are chordal, the
results concerning upper bounds also hold for threshold graphs, see Table \ref{summary}. 
While counting all maximal independent set in chordal
graphs is $\#\p$-complete \cite{OUU08}, for threshold graphs the problem
is easy by Corollary \ref{cor-enum2}.

%%%%%%%%%%%%%%%%%%%%%%%%%%%%%%%%%%%%%%%%%%%%%%%%%%%%%%%%%%%%%%%%%%%%%%
\subsubsection{Solutions for the Knapsack Problem}
%%%%%%%%%%%%%%%%%%%%%%%%%%%%%%%%%%%%%%%%%%%%%%%%%%%%%%%%%%%%%%%%%%%%%%

\begin{theorem}\label{main-th-kp}
Let $I$ be an instance for {\sc Max KP} on $n$ items which has an equivalent graph.
Then $I$ can be solved in time $\bigo(n^2)$.
\end{theorem}

\begin{proof}
Let $I$ be some instance for {\sc Max KP}  
on $n$ items which has an equivalent graph.
By Lemma \ref{le-k-gr} instance $I$ is equivalent
to graph $G(I)$, which can be constructed from $I$ in time  $\bigo(n^2)$.
Graph $G(I)$ is a threshold graph by
Observation \ref{obs-g-i}.
%A creation sequence for $G(I)$ can be found in time  $\bigo(n+m)$
%using the method in Figure 1.4 of \cite{MP95a}.

Thus the $\omega(G(I))\leq n$ maximal independent sets in $G(I)$ 
can be found in $\bigo(n^2)$ by Corollary \ref{cor-enum2} and correspond to the 
maximal feasible solutions of $I$. For every of these solutions
we can compute its profit in time $\bigo(n)$.
\end{proof}

Since there are several instances  for {\sc Max KP} which do not allow
an equivalent graph (cf.\ Example \ref{ex-no-kp-i}) this does not
imply that we can solve all instances for {\sc Max KP} in polynomial time.

Knapsack problems have also been studied in connection with conflict graphs \cite{PS09}.
The solution of knapsack problems with conflict graphs (KPC) are independent
sets in the conflict graph.
%it seems to be possible to extend the approaches
%for KP instances with knapsack graphs to KPC instances with knapsack graphs.  
For example for chordal conflict graphs  a pseudo-polynomial solution is known from \cite{PS09}.
%which we could be improved for the subclasses threshold graphs and split graphs.
Since maximal independent sets in the conflict graph do not necessarily
correspond to feasible solutions in the knapsack instance, our methods
do not imply solutions for  knapsack problems with conflict graphs.

%In Section 4  of \cite{CLS81} results on the counting and enumeration problems
%on co-graphs are given, which also hold for threshold graphs.

%%%%%%%%%%%%%%%%%%%%%%%%%%%%%%%%%%%%%%%%%%%%%%%%%%%%%%%%%%%%%%%%%%%%%%%%%%%
\subsection{Counting and enumerating  maximum independent sets}\label{th-count-maxis}
%%%%%%%%%%%%%%%%%%%%%%%%%%%%%%%%%%%%%%%%%%%%%%%%%%%%%%%%%%%%%%%%%%%%%%%%%%%

Since every maximum independent set is a maximal independent set,
our results given in Section \ref{ce-mis-th} also can be applied to list all maximum independent sets
in threshold graphs. By 
the method given in Figure \ref{fig:mis_thres} and removing 
non-maximum sets we obtain a
method for listing 
all maximum independent sets
in a threshold graph.

\begin{corollary}\label{cor-enum2-}
All maximum independent sets in a threshold graph  on $n$ vertices and $m$
edges can be counted 
and enumerated in time $\bigo(\omega(G)\cdot n+m)$.
\end{corollary}

Next we give a method to enumerate all maximum independent sets in a 
threshold graph.

\begin{figure}[ht]
\hrule
\medskip
\begin{tabbing}
xxx \= xxx \= xxx \= xxx \= xxx\= xxx \= xxx \= xxx\= xxx \= xxx\kill
$\IM(G)=\emptyset$; \\
$j=1$;
while ($t_j=1$) $j++$; \>\>\>\>\>\>\>\>\>  $\vartriangleright$ find first $0$ in $t$ if exists \\
%let $j=\min\{i~|~t_i=0\}$\\
let $A=\{v(i)~|~t_i=0\}$\\
for ($i = 1$; $i<j$; $i++$) \\
\>  $\IM(G)= \IM(G) \cup \{\{v(i)\} \cup A\}$; 
\end{tabbing}
\hrule
\caption{Enumerating all maximum   independent sets in a threshold graph.}
\label{fig:km_is_thres}
\end{figure}

\begin{theorem}\label{cor-enum-im} 
Let $G$ be a threshold graph $G$ on $n$ vertices 
which is given by a creation sequence. Then all  maximum independent sets in $G$
can be counted in time $\bigo(n)$ and 
enumerated  in time $\bigo(\omega(G)\cdot \alpha(G))$.
\end{theorem}

\begin{proof}
The number of maximum independent sets 
is equal to the value of $j-1\leq \omega(G)$ (cf.\ Figure \ref{fig:km_is_thres}) and 
each set has $\alpha(G)$ elements. 
\end{proof}

Since a creation sequence can be found in time  $\bigo(n+m)$
by Lemma \ref{find-cs} we obtain the following result.

\begin{corollary}\label{cor-enum2y}
Let $G$ be a threshold graph on $n$ vertices and $m$ edges.  
The number of all  maximum independent sets in $G$ equals $\omega(G)$.
All maximum independent sets in $G$ can be counted in time $\bigo(n+m)$.  
All maximum independent sets in $G$  can be enumerated  in time 
$\bigo(\omega(G)\cdot n+m)$. 
After $\bigo(n+m)$ time precomputation all independent sets in $G$  can be enumerated 
in constant time per output.
\end{corollary}

The latter result can also be obtained from the fact that threshold graphs are
chordal graphs and for these graphs in \cite{OUU08} it has been shown
that  all maximum independent sets in a chordal graph  on $n$ vertices and $m$
edges can be counted in time  $\bigo(n+m)$
and enumerated in time $\bigo(1)$ per output.
The related problem of finding {\em one} maximum
independent set in a  threshold graph
was solved in \cite{CLR04} in time $\bigo(n\log n)$.
This problem can also be solved by Corollary \ref{find-maxcln} or by Remark \ref{tresh-spli-p}.
A remarkable difference between our solutions and these of \cite{CLR04}
is that we work with graph representations and the authors of \cite{CLR04}
use the coefficients occurring in the knapsack instance. Each of these
versions can transformed into the other in quadratic time
by Lemma \ref{le-gr-k} and \ref{le-k-gr}.

%%%%%%%%%%%%%%%%%%%%%%%%%%%%%%%%%%%%%%%%%%%%%%%%%%%%%%%%%%%%%%%%%%%%%%%%%%%
\subsection{Counting and enumerating independent sets}\label{th-count-is}
%%%%%%%%%%%%%%%%%%%%%%%%%%%%%%%%%%%%%%%%%%%%%%%%%%%%%%%%%%%%%%%%%%%%%%%%%%%

Next we want to enumerate and count  all independent sets in a 
threshold graph. While the number of maximal and maximum independent sets
is bounded by the number $n$ of vertices the number of 
independent sets can be in $\Theta(2^n)$, e.g. 
for edgeless graphs or stars on $n$ vertices.

\begin{figure}[ht]
\hrule
\medskip
\begin{tabbing}
xxx \= xxx \= xxx \= xxx \= xxx\= xxx \= xxx \kill
$\IS(G)=\emptyset$; \\
for ($i = 1$; $i\leq n$; $i++$) \\
\>    if ($t_i=1$) \\
\> \> $\IS(G)= \IS(G) \cup \{\{v(i)\}\}$;  \\
\>   else   \>\>\>\>  $\vartriangleright$ $t_i=0$\\
\> \> $\IS(G)= \IS(G) \cup \{\{v(i)\}\} \cup \{\{I\cup \{v(i)\}~|~ I \in \IS(G) \}\}$; 
\end{tabbing}
\hrule
\caption{Enumerating all  independent sets in a threshold graph.}
\label{fig:is_thres}
\end{figure}

\begin{theorem}\label{t-i-enum2}
Let $G$ be a threshold graph on $n$ vertices which is  
given by a creation sequence.  
All independent sets in $G$ can be counted in time $\bigo(n)$.  
All independent sets in $G$  can be enumerated  in time 
$\bigo(n\cdot 2^{n-1})$.
\end{theorem}

\begin{proof}
Let $G$ be given by creation sequence $t=t_1\ldots t_n$ and
$G_i$ be the subgraph of $G$ induced by $\{v(1),\ldots,v(i)\}$.
Then $\is(G_1)=1$. For $i>1$ and for $t_i=1$ we have $\is(G_i)=\is(G_{i-1})+1$
and  for $t_i=0$ we have $\is(G_i)=2\cdot \is(G_{i-1})+1$. Thus we
can count all  independent sets in $G$ in time  $\bigo(n)$.
The time for listing all  independent sets in $G$ depends on
the size and this is at most the size $\sum_{i=1}^n i \cdot \binom{n}{i}=n\cdot 2^{n-1}$ 
of the independent sets in an edgeless graph.
\end{proof}

%%%%%%%%%%%%%%%%%%%%%%%%%%%%%%%%%%%%%%%%%%%%%%%
% die maximale anzahl und die Site für I_n:
% 2^n Mengen
% 2^{n-1} n  Größe, n Mengen der gröss 1, binon(n,2) Mengen der größe 2, ... usw.
%%%%%%%%%%%%%%%%%%%%%%%%%%%%%%%%%%%%%%%%%%%%%%%

Since a creation sequence can be found in time  $\bigo(n+m)$
by Lemma \ref{find-cs} we obtain the following result.

\begin{corollary}\label{cor-enum-is}
Let $G$ be a threshold graph on $n$ vertices and $m$ edges.  
All independent sets in $G$ can be counted in time $\bigo(n+m)$.  
All independent sets in $G$  can be enumerated  in time 
$\bigo(n\cdot 2^{n-1})$. After $\bigo(n+m)$ time
precomputation all independent sets in $G$  can be enumerated 
in constant time per output.
\end{corollary}

%%%%%%%%%%%%%%%%%%%%%%%%%%%%%%%%%%%%%%%%%%%%%%%%%%%%%%%%%%%%%%%%%%%%%%%%%%%
\subsection{Bin Packing}\label{sec-bp}
%%%%%%%%%%%%%%%%%%%%%%%%%%%%%%%%%%%%%%%%%%%%%%%%%%%%%%%%%%%%%%%%%%%%%%%%%%%

In \cite{CLR04} threshold graphs were briefly considered in order to 
give lower bounds on solutions for bin packing problems.\footnote{Please 
note the different but equivalent definition of edges within thres\-hold graphs
in \cite{CLR04}. By using values  $w_i$ for the vertices $v_i$
such that $0<w_i\leq 1$ they define edges  $\{v_i,v_j\}$ whenever
$w_i+w_j\leq 1$. For our notations see Definition \ref{th-thres} (\ref{c-edge}).}
But in \cite{CLR04} it is not discussed whether the instances
of  bin packing problems posses equivalent graphs.

\begin{desctight}
\item[Name] {\sc Min Bin Packing} ({\sc Min BP})

\item[Instance] A set $A=\{a_1,\ldots,a_n\}$ of $n$ items
and a number $n$ of bins. For every item $a_j$, there 
is a positive rational size of $s_j$. 

\item[Task] 
Find $n$ disjoint (possibly empty) subsets $A_1,\ldots, A_n$ of $A$ such that 
the  number of non-empty subsets is minimized 
and the sum of the sizes of the items in each subset is at most $1$.
\end{desctight}

For some instance $I$ for {\sc Min BP}
two items $a_j$ and $a_{j'}$ can be chosen into the same subset (bin) if and only if
\begin{equation}
s_{j}+s_{j'}\leq 1.\label{gl-smaller1}
\end{equation}
This motivates to model the compatibleness of the items in 
$I$ by  threshold graphs.
Therefore we consider the inequality 
\begin{equation}
s_{1}x_1+s_{2}x_2+ \ldots + s_{n}x_n \leq 1. \label{gl1}
\end{equation}

An instance $I$ for {\sc Min BP} and a graph $G=(V,E)$
are {\em equivalent}, if there is a bijection $f:A\to V$ such that for every $A'\subseteq A$ 
the characteristic vector of $A'$
satisfies inequality  (\ref{gl1})
if and only if $f(A'):=\{f(a_j)~|~a_j\in A'\}$ is an independent 
set of $G$.

If we assume that instance $I$ 
has an equivalent  graph $G=(V,E)$, then 
we know by condition (\ref{c1}) of Theorem \ref{th-thres} 
that $G$ is a threshold graph.
Further by (\ref{gl-smaller1}) and (\ref{gl1}) two items $a_j$ and $a_{j'}$ can be chosen into 
the same bin if and only if $\{v_j,v_{j'}\}\not\in E$.
That is any two items
corresponding to the vertices of a clique in $G$ can not be chosen
into the same bin. Thus $\omega(G)$ leads a lower bound on the number of 
bins needed for the considered 
instance $I$ for {\sc Min BP}. The value $\omega(G)$ can  be
computed in linear time from threshold graph $G$, see Corollary \ref{find-maxcln}.

Let $\Pi$ be some optimization problem and $I$ be some
instance of $\Pi$. By $OPT(I)$ we denote the value
of an optimal solution for $\Pi$ on input $I$. 

\begin{theorem}
Let $I$ be an instance for {\sc Min BP} which 
has an equivalent graph $G$. Then $OPT(I)\geq \omega(G)$.
\end{theorem}

In order to know an equivalent graph we can proceed as in Section \ref{sec-graph-fr-pr}
and use graph $G(I)$ defined by the values of inequality (\ref{gl1}).

\begin{theorem}
Let $I$ be an instance for {\sc Min BP} which
has an equivalent graph. Then $OPT(I)\geq \omega(G(I))$.
\end{theorem}

%%%%%%%%%%%%%%%%%%%%%%%%%%%%%%%%%%%%%%%%%%%%%%%%%%%%%%%%%%%%%%%%%%%%%%%%%%%
\section{Multidimensional Knapsack Problem and $k$-Threshold Graphs}\label{sec-d-knaps}
%%%%%%%%%%%%%%%%%%%%%%%%%%%%%%%%%%%%%%%%%%%%%%%%%%%%%%%%%%%%%%%%%%%%%%%%%%%

%%%%%%%%%%%%%%%%%%%%%%%%%%%%%%%%%%%%%%%%%%%%%%%%%%%%%%%%%%%%%%%%%%%%%%%%%%%
\subsection{Multidimensional Knapsack Problem}
%%%%%%%%%%%%%%%%%%%%%%%%%%%%%%%%%%%%%%%%%%%%%%%%%%%%%%%%%%%%%%%%%%%%%%%%%%%

Next we consider the knapsack problem for multiple dimensions, see \cite{KPP10} 
for a survey.

\begin{desctight}
\item[Name] {\sc Max d-dimensional Knapsack} ({\sc Max d-KP})

\item[Instance] A set $A=\{a_1,\ldots,a_n\}$ of $n$ items and a number $d$ of dimensions.
Every item $a_j$ has a profit $p_j$ and for dimension $i$ the size $s_{i,j}$. 
Further for every dimension $i$ there is a capacity $c_i$.

\item[Task]  Find a subset $A'\subseteq A$ such that 
the total profit of $A'$  is maximized %($\sum_{a\in A'}$)
and 
for every dimension $i$ the total size of $A'$ is at most the capacity $c_i$.
\end{desctight}

In the case of $d=1$ the {\sc Max d-KP} problem corresponds to the {\sc Max KP} problem considered in
Section  \ref{sec-knaps}.
The parameters $n$ and $d$ are assumed to be
positive integers and  $p_j$, $s_{i,j}$, and $c_i$ are assumed to be non-negative reals.
%\footnote{In some 
%works it is allowed to have  $s_{i,j}=0$ for 
%some $1\leq i \leq d$, $1\leq j\leq n$  
%if $\sum_{i=1}^{d}s_{i,j}\geq 1$ for every item $a_j$, see \cite[Chapter 9]{KPP10}}
Let $I$ be an instance for {\sc Max d-KP}.
Every subset $A'$ of $A$ such that  $\sum_{a_j\in A'}s_{i,j}\leq c_i$ for every $i\in[d]$
is a {\em feasible solution} of $I$. 
A feasible solution
which is not the subset of another feasible solution is called {\em maximal}. 

\begin{definition}\label{def-eq-dk}
An instance $I$ for {\sc Max d-KP} and a graph $G=(V,E)$
are {\em equivalent}, if there is a bijection $f:A\to V$ such that $A'\subseteq A$
is a feasible solution of $I$ if and only if $f(A'):=\{f(a_j)~|~a_j\in A'\}$ is an independent 
set of $G$.
\end{definition}

%It remains to consider maximal feasible solutions and maximal independent sets 
%(cf.\ proof of Lemma \ref{max-sol-kp}).
%
%\begin{lemma}\label{max-sol-d-kp}
%Let $I$ be an  instance for {\sc Max d-KP} and $G=(V,E)$ be a graph. If there is 
%a bijection $f:A\to V$ such that $A'\subseteq A$
%is a maximal feasible solution of $I$ if and only if $\{f(a_j)~|~a_j\in A'\}$ is a 
%maximal independent set of $G$, then $I$ and $G$ are equivalent.
%\end{lemma}

%%%%%%%%%%%%%%%%%%%%%%%%%%%%%%%%%%%%%%%%%%%%%%%%%%%%%%%%%%%%%%%%%%%%%%%%%%%
\subsection{$k$-Threshold Graphs}
%%%%%%%%%%%%%%%%%%%%%%%%%%%%%%%%%%%%%%%%%%%%%%%%%%%%%%%%%%%%%%%%%%%%%%%%%%%

In order to characterize graphs which are equivalent to instanced  for {\sc Max d-KP}, we recall
the notation of the threshold dimension from \cite{Gol80}.
The {\em threshold dimension} of a graph $G=(V,E)$ on $n$ vertices, $t(G)$ for short, 
is the minimum number $k$ of
linear inequalities
\begin{equation}
a_{i,1}x_1 + \ldots + a_{i,n}x_n \leq  T_i\label{eq}
\end{equation}
such that $X\subseteq V$ is an independent set in $G$ if and only if the 
characteristic vector $(x_1,\ldots,x_n)$ of $X$ satisfies for $i=1,\ldots,k$ the inequalities of type (\ref{eq}).
By \cite[page 221]{Gol80} we can assume that all $a_{i,j}$ and $T_i$ are non-negative
integers. 

\begin{theorem}[\cite{CH77}]\label{th-k-thres}
For every graph $G=(V,E)$ the following statements are equivalent.
\begin{enumerate}

\item\label{ck1}  $G$ has threshold dimension at most $k$.

\item There are at most $k$ of threshold graphs $G_i=(V,E_i)$ such
that $E=E_1\cup \ldots \cup E_k$.
\end{enumerate}
\end{theorem}

Graphs $G$ such that  $t(G)=0$ are exactly the edgeless graphs and
graphs  $G$ such that  $t(G)\leq 1$ are exactly threshold graphs.
A complete characterization for graphs $G$ such that  $t(G)\leq 2$
is unknown. 

A graph is denoted as {\em $k$-threshold} if and only if its threshold dimension is at most $k$.

%There is a close relation between the  threshold dimension and the 
%independence number (cf.\ \cite{CH73}) or clique number.
%
%
%\begin{lemma}Let $G$ be some graph on $n$ vertices, then 
%$t(G)\leq n- \max\{\omega(G)-1,\alpha(G)\}$.
%\end{lemma}
%%
%
%\begin{proof}
%Let $G$ be some graph. Next we define a covering of its edge set with 
%$n- \max\{\omega(G)-1,\alpha(G)\}$
%threshold graphs. Let $V'$ be the vertex set
%of a maximum clique (maximum independent set). 
%Then $G[V']$ is a threshold graph.
%For each of the remaining vertices $v\in V-V'$, we insert
%a star consisting of $v$, its neighbours in $G$, and edges between $v$ 
%to its neighbours in $G$. Each of these graphs is a threshold graph.
%Since for an independent set $V'$ the graph $G[V']$ is
%edgeless and can be omitted in the covering, we obtain a covering by 
%at most
%$n- \max\{\omega(G)-1,\alpha(G)\}$ threshold graphs. \qed
%\end{proof}

For $k=1$ the set of $k$-threshold graphs is closed under edge complements (Lemma \ref{comple}).
This does not hold for $k\geq 2$ since $t(K_{n,n})=n\neq 2=t(\overline{K_{n,n}})$ for $n\geq 3$.
That is, in general the threshold dimension of a graph $G$ is not equal to the threshold dimension
of its edge complement $\overline{G}$.
But there is an interesting relation between the threshold dimension and
the threshold intersection dimension of the edge complement graph. For 
some graph $G$ its {\em threshold intersection dimension}, $ti(G)$ for short, 
is the minimum number of threshold graphs whose intersection is $G$ \cite{Riv85}.
We (did not find some assigned notation and) 
call a graph  {\em $k$-threshold intersection} if and only if its threshold intersection 
dimension is at most $k$.

\begin{lemma}
For every graph $G$ it holds $t(G)=ti(\overline{G})$.
\end{lemma}
%page 310 of \cite{Riv85}

\begin{proof}
The complement of a graph covered by $k$ graphs $G_1,\ldots,G_k$ is the
intersection of the complements. 
If the graphs $G_1,\ldots,G_k$ are threshold graphs then their
complements are also threshold graphs.
\end{proof}

\begin{lemma}
For every graph $G$ the threshold dimension $t(G)$ is the minimum number $k$ such
that $\overline{G}$ is the intersection of $k$ threshold graphs.
\end{lemma}

Threshold dimension 1 can be decided in linear time 
and threshold dimension 2 can be decided in polynomial time \cite{SR98}.

\begin{theorem}[\cite{Yan82}]\label{hard} For every graph $G$ it is NP-complete to
determine whether $t(G)\leq k$ where $k\geq 3$.
\end{theorem}

The number of $k$-threshold graphs can be bounded by the union
of one of $2^{n-1}$ (cf.\ Observation \ref{number}) threshold
graphs for every of $k$ dimensions.

\begin{observation}\label{k-number}
There are at most $2^{k(n-1)}$
%$\binom{2^{n-1}+k-1}{k}\leq (2^{n-1}+1)^k$ 
many $k$-threshold 
graphs on $n$ vertices.
\end{observation}

%%%%%%%%%%%%%%%%%%%%%%%%%%%%%%%%%%%%%%%%%%%%%%%%%%%%%%%%%%%%%%%%%%%%%%
\subsection{Creating multidimensional knapsack problems from $k$-threshold graphs}
%%%%%%%%%%%%%%%%%%%%%%%%%%%%%%%%%%%%%%%%%%%%%%%%%%%%%%%%%%%%%%%%%%%%%%

Let $G=(V,E)$ be a $k$-threshold graph.
Then we can define an equivalent instance for {\sc Max d-KP}  
as follows. Because of the hardness mentioned in Theorem \ref{hard}
we assume that we have given $k$ threshold graphs
which cover the edge set of $G$.

\begin{lemma}\label{le-gr-k2}
Given is some $k$-threshold graph $G=(V,E)$ on $n$ vertices which is given 
by $k$ threshold graphs $G_j=(V,E_j)$, $1\leq j \leq k$, which
cover the edge set of $G$. An equivalent instance $I$ for {\sc Max d-KP} can be 
constructed in
time $\bigo(d\cdot n^2)$.
\end{lemma}

\begin{proof}
We assume that we have given $k$ threshold graphs $G_j=(V,E_j)$, $1\leq j \leq k$, which
cover the edge set of $G$. For each of these graphs $G_j$ we can find a
creation sequence in $\bigo(n^2)$. Then we define by the method given in Figure
\ref{fig:thres_to_knap} sizes $s_{j,i}$ and a capacity $c_j$  to obtain an
instance for  {\sc Max KP}. The union of all these $k$ instances (inequalities)
is an equivalent  instance for {\sc Max d-KP}. 
\end{proof}

%%%%%%%%%%%%%%%%%%%%%%%%%%%%%%%%%%%%%%%%%%%%%%%%%%%%%%%%%%%%%%%%%%%%%%
\subsection{Creating $k$-threshold graphs from multidimensional knapsack problems}\label{sec-k-t-f-mkp}
%%%%%%%%%%%%%%%%%%%%%%%%%%%%%%%%%%%%%%%%%%%%%%%%%%%%%%%%%%%%%%%%%%%%%%

Next we consider the problem of
defining a graph from an instance for {\sc Max d-KP}.
First we look at the relation between the existence of an equivalent graph 
for some  {\sc Max d-KP}  instance $I$
and the existence of an equivalent graph for
the $d$ corresponding  {\sc Max KP}  instances, which are defined as follows.
Let $I$ be some {\sc Max d-KP}  instance on item set $A=\{a_1,\ldots,a_n\}$
with sizes $s_{1,1},\ldots,s_{d,n}$, capacities $c_1,\ldots,c_d$ and $i\in[d]$.
Then for $i\in [d]$ we define by $I_i$ the {\sc Max KP}  instance on item set $A=\{a_1,\ldots,a_n\}$
with profits $p_1,\ldots, p_n$, sizes $s_{i,1},\ldots,s_{i,n}$, and capacity $c_i$.

\begin{lemma}\label{cor-kp}
Let $I$ be some instance of {\sc Max d-KP}. 
If for every dimension $i\in[d]$ the corresponding  {\sc Max KP} instance $I_i$
possesses an equivalent graph $G_i=(V,E_i)$, then $G=(V,E)$ where $E=E_1\cup \ldots \cup E_d$ leads an 
equivalent graph for instance $I$.
\end{lemma}

\begin{proof}
Let $I$ be some instance of {\sc Max d-KP} on item set $A=\{a_1,\ldots,a_n\}$
with sizes $s_{1,1},\ldots,s_{d,n}$, capacities $c_1,\ldots,c_d$ and $i\in[d]$.
and  $I_1,\ldots,I_d$
the  corresponding  {\sc Max KP}  instances.
Assume that for every dimension $i\in[d]$  {\sc Max KP} instance $I_i$
possesses an equivalent graph $G_i$. Every $G_i$ is  is a threshold graph by
definition. Since all graphs
are defined on the item set of $I$ they can be defined on the
same vertex set, i.e. $G_i=(V,E_i)$. 
%(We even know that $G_i$ is isomorphic to $G(I_i)$.)
Then we obtain a $d$-threshold graph
$G=(V,E)$ by $E=E_1\cup \ldots \cup E_d$.
Then $A'\subseteq A$ is a feasible solution for $I$ if and only if
$A'$  is a feasible solution for every $I_i$ if and only if
$f(A')$ is an independent set of of every $G_i$ if and only if
$f(A')$ is an independent set of  $G$.
Thus $G$ is equivalent to $I$.
\end{proof}

On the other hand, it is not possible to carry over the existence
of an equivalent graph from $I$ to the instances $I_1,\ldots,I_d$.

\begin{example}\label{ex-no-kp-i3} %no + yes = yes
We consider $d=2$ dimensions. First let
$s_{1,1}=3$, $s_{1,2}=1$, $s_{1,3}=2$, $s_{1,4}=4$, $s_{1,5}=5$ and $c_1=5$. 
Further let 
$s_{2,1}=5$, $s_{2,2}=5$, $s_{2,3}=5$, $s_{2,4}=1$, $s_{2,5}=1$ and $c_2=5$.
Then $I_1$ does not allow an equivalent graph
and $I_2$ allows an equivalent graph ($K_5-e$). Instance $I$
allows an equivalent graph ($K_5$).
\end{example}

That is, if for some dimension $i$ the instance $I_i$ 
does not allow an  equivalent graph, this not necessary implies
that $I$ does not have an equivalent graph. But this assumption
also can imply that 
$I$ does not have an equivalent graph by the following
example.

\begin{example}\label{ex-no-kp-i2} % not + yes = not
We consider $d=2$ dimensions. First let
$s_{1,1}=12$, $s_{1,2}=10$, $s_{1,3}=11$, $s_{1,4}=8$, $s_{1,5}=9$ and $c_1=26$. 
Further let $s_{2,1}=2$, $s_{2,2}=1$, $s_{2,3}=2$, $s_{2,4}=4$, $s_{2,5}=5$ and $c_2=5$.
Then $I_1$ does not allow an equivalent graph.
and $I_2$ allows an equivalent graph (\gansfuss{dart} graph). Instance $I$
does not allow an equivalent graph, since for $A'=\{a_1,a_2,a_3\}$
we every two different items of $A'$ form a feasible solution of $I$ 
which would imply an independent set of size three in a corresponding
graph, but $A'$ itself is not a feasible solution.
\end{example}

If none of the instances $I_1,\ldots,I_d$ has an equivalent graph, for
instance $I$ both situations are possible. By the following two examples.

\begin{example}\label{ex-no-kp-i4} %no + no = yes
We consider $d=2$ dimensions. First let
$s_{1,1}=3$, $s_{1,2}=1$, $s_{1,3}=2$, $s_{1,4}=5$, $s_{1,5}=5$,  $s_{1,6}=5$, and $c_1=5$. 
Further let 
$s_{2,1}=5$, $s_{2,2}=5$, $s_{2,3}=5$, $s_{2,4}=3$, $s_{2,5}=1$,  $s_{1,6}=2$,  and $c_2=5$.
Then $I_1$ and $I_2$ do not allow an equivalent graph. Instance $I$
allows an equivalent graph ($K_6$).
\end{example}

\begin{example}\label{ex-no-kp-i4n} %no + no = no
We consider $d=2$ dimensions. First let
$s_{1,1}=3$, $s_{1,2}=1$, $s_{1,3}=2$, $s_{1,4}=5$, $s_{1,5}=5$,  $s_{1,6}=5$, and $c_1=5$. 
Further let 
$s_{2,1}=3$, $s_{2,2}=1$, $s_{2,3}=2$, $s_{2,4}=3$, $s_{2,5}=1$,  $s_{1,6}=2$,  and $c_2=5$.
Then $I_1$ and $I_2$ do not allow an equivalent graph. But now also instance $I$
also does not allow an equivalent graph.
\end{example}

Next we want to characterize {\sc Max d-KP} instances allowing an equivalent
graph. Therefor let $ \mathcal{I}_d$ be the set of all instances  of {\sc Max d-KP}. 
We define a Boolean property $P: \mathcal{I}_d \to \{\text{true},\text{false}\}$ for 
some instance $I\in \mathcal{I}_d$ by
\begin{equation}
P_d(I)=\lfa_{A'\subseteq A}\big((\lfa_{i\in [d]}\lfa_{a_{i,j},a_{i,j'}\in A'} s_{i,j}+s_{i,j'}\leq c_i) \Rightarrow (\lfa_{i\in [d]}\sum_{a_j\in A'}s_{i,j}\leq c_i) \big).
\end{equation}

The idea of this property is to ensure feasibility in the case of independence
within $A$. Although independence has to be valid within a graph, this will be useful.
For subsets $A'$ with $|A'|\leq 2$ the property is always true.
Since we assume that $j\neq j'$ the implication 
from the right to the left is always true. Comparing property $P_d$
with property $P$ defined in Section \ref{sec-graph-fr-pr} we conclude
that if property $P(I_i)$ holds for every $i\in[d]$
then also property $P_d(I)$ holds, but the reverse direction does not hold in general
(see Examples \ref{ex-no-kp-i3} and \ref{ex-no-kp-i4}).

\begin{lemma}\label{lemma-g-pd}
If some instance $I$ for {\sc Max d-KP} has an equivalent
graph, then $I$ satisfies property $P_d(I)$.
\end{lemma}

\begin{proof}
Let $I$ be some instance for {\sc Max d-KP}. In order to give a proof by
contradiction we assume that $I$ does not satisfy property $P_d(I)$.
That is there is some $A'\subseteq A$ such that for every 
$i\in[d]$ and $\forall_{a_j,a_{j'}\in A'} s_{i,j}+s_{i,j'}\leq c_i$ 
and there is some dimension $i'\in [d]$ such that $\sum_{a_j\in A'}s_{i',j}> c_{i'}$. 
Then there is some set $A''\subseteq A'$, $a_\ell\in A''$ and $i''\in [d]$ such 
that $\sum_{a_j\in A''}s_{i'',j}> c_{i''}$ and for all 
$i\in[d]$ it holds $\sum_{a_j\in A''-\{a_\ell\}}s_{i,j}\leq c_{i}$. 
That means that for every $a_j\in A''$ set $\{a_j,a_\ell\}$ is a feasible solution for $I$ 
and $A''-\{a_\ell\}$ is also a feasible solution for $I$, but 
$A''$ is not a feasible solution for $I$.

Every bijection $f:A\to V$ certifying an equivalent graph $G=(V,E)$ for $I$ has to map 
for every  $a_j\in A''$ item set $\{a_j,a_\ell\}$ onto an independent set $\{f(a_j),f(a_\ell)\}$
and item  set $A''-\{a_\ell\}$  onto an independent set $f(A''-\{a_\ell\})=f(A'')-f(\{a_\ell\})$.
But then also $f(A'')$ is an independent set in $G$ while $A''$ 
is not a feasible solution for $I$.
\end{proof}

Next we want to show the reverse direction of Lemma \ref{lemma-g-pd}. For some 
instance $I$ of {\sc Max d-KP} we define the graph $G_d(I)=(V(I),E(I))$ by
%$$
%\begin{array}{lcl}
%V(I)&=&\{v_j~|~a_j\in A\}  \text{ and} \\
%E(I)&=&\{\{v_{j},v_{j'}\}~|~ s_{1,j}+s_{1,j'}> c_1 \vee \ldots \vee s_{d,j}+s_{d,j'}> c_d \}.
%\end{array}
%$$
\begin{equation}
V(I)=\{v_j~|~a_j\in A\} ~~~ \text{ and } ~~~  E(I)=\{\{v_{j},v_{j'}\}~|~ s_{1,j}+s_{1,j'}> c_1 \vee \ldots \vee s_{d,j}+s_{d,j'}> c_d \}.
\end{equation}

In general $G(I)$ is not equivalent to $I$, see Example \ref{ex-no-kp-i2}.
Graph $G_d(I)$ can also be obtained from the $d$ threshold
graphs $G(I_i)=(V(I_i),E(I_i))$, $1\leq i \leq d$, by $E(I)=E(I_1)\cup \ldots \cup
E(I_d)$, which implies that $G_d(I)$ is a $d$-threshold graph by Theorem \ref{th-k-thres}.

\begin{observation}\label{obs-g-i-d}
For every  instance $I$ for {\sc Max d-KP} graph $G_d(I)$ is a $d$-threshold graph.
\end{observation}

Next we give a tight connection between the property $P_d(I)$ and
graph $G_d(I)$. Since both are defined by the sizes of $I$, the
result follows straightforward.

\begin{lemma}\label{lemma-p-g-d}
Let $I$ be some instance $I$ for {\sc Max d-KP}. Then $I$  satisfies property $P_d(I)$, 
if and only if $I$ is equivalent to graph $G_d(I)$.
\end{lemma}

\begin{proof}
Let $I$ be an instance for {\sc Max d-KP}. 

First assume that $I$ satisfies property $P_d(I)$.
Let $A'\subseteq A$ be a feasible solution for $I$, i.e. for every $i\in[d]$ 
it holds $\sum_{a_j\in A'}s_{i,j}\leq c_i$.
Then obviously also for every $i\in[d]$ it holds $\forall_{a_j,a_{j'}\in A'} s_{i,j}+s_{i,j'}\leq c_i$ holds 
and by the definition of edge set $E(I)$ the vertex set 
$V'=\{v_j ~|~ a_j\in A'\}$ is an independent set in graph $G_d(I)$.
For the reverse direction let
$V'\subseteq V$ be an independent set in $G_d(I)$. 
Then for $A'=\{a_j~|~v_j\in V'\}$  for every $i\in[d]$
it holds $\forall_{a_{j},a_{j'}\in A'} s_{i,j}+s_{i,j'}\leq c_i$ and since
property $P_d(I)$ is satisfied this implies that  for every $i\in[d]$ it holds
$\sum_{a_j\in A'}s_{i,j}\leq c_i$, i.e. 
$A'$ is a feasible solution for $I$. That is, the bijection defined 
by the definition of $V(I)=\{v_j~|~a_j\in A\}$ verifies that $I$ is equivalent to $G(I)$.        

Next assume that  $I$ is equivalent to graph $G_d(I)$. Let $f$ be the bijection between
the item set $A$ and the vertex set $V(I)$, which exists by Definition \ref{def-eq-dk}.
Let $A'\subseteq A$. 
If for every $i\in[d]$ it holds 
$\forall_{a_j,a_{j'}\in A'} s_{i,j}+s_{i,j'}\leq c_i$, then by the definition
of $E(I)$ for all $a_j,a_{j'}\in A'$ it holds $\{f(a_j),f(a_{j'})\}\not\in E$ and thus
set $V'=\{f(a_j) ~|~ a_j\in A'\}$ is an independent set of $G_d(I)$, which
corresponds to a feasible  solution $A'$ by the equivalence of $I$ and $G_d(I)$ via 
bijection $f$. Thus for every $i\in[d]$ it holds 
$\sum_{a_j\in A'}s_{i,j}\leq c_i$, which  implies that $P_d(I)$ holds true.
\end{proof}

Now we can state a result corresponding to the reverse direction of Lemma \ref{le-gr-k2}.

\begin{lemma}\label{le-k-gr2}
Let $I$ be some instance for {\sc Max d-KP} on $n$ items which has an equivalent graph. Then
$I$ is equivalent to graph $G_d(I)$, which  can be constructed from $I$ in time $\bigo(d\cdot n^2)$.
\end{lemma}

\begin{proof}
Let $I$ be some instance for {\sc Max d-KP} on $n$ items which has an equivalent graph.
By Lemma \ref{lemma-g-pd} we know that $I$ satisfies property $P_d(I)$ and by Lemma \ref{lemma-p-g-d}
we know that
$I$ is equivalent to graph $G_d(I)$. Furthermore we
can construct $G_d(I)$  in time $\bigo(d\cdot n^2)$
by its definition.
\end{proof}

We obtain the following characterizations
for instances $I$ of {\sc Max d-KP} allowing an equivalent graph.

\begin{theorem}\label{corollary-char-d}
For every  instance $I$ of {\sc Max d-KP}  the
following conditions are equivalent.
\begin{enumerate}
\item Instance $I$ has an equivalent graph.
\item Instance $I$ satisfies property $P_d(I)$.
\item Instance $I$ is equivalent to graph $G_d(I)$.
\item Instance $I$ has an equivalent $d$-threshold graph.
\end{enumerate}
\end{theorem}

\begin{proof}
(1) $\Rightarrow$(2)  by Lemma \ref{lemma-g-pd},      (2) $\Leftrightarrow$ (3)  by Lemma \ref{lemma-p-g-d}, 
(3) $\Rightarrow$(4)  by Observation \ref{obs-g-i-d}, (4) $\Rightarrow$     (1) obvious
\end{proof}

The proof of the following result runs similar to that of Corollary \ref{corollary-char-gi}.

\begin{corollary}\label{corollary-char-gi-d}
Let $I$ be some instance for {\sc Max d-KP}  which 
has two equivalent graphs $G_1$ and $G_2$.
Then $G_1$ is isomorphic to $G_2$.
\end{corollary}

Whenever  some instance $I$ for {\sc Max d-KP} has an equivalent graph, then by Theorem \ref{corollary-char-d}
graph $G_d(I)$ is also an equivalent graph for $I$  thus
we obtain the next result.

\begin{corollary}\label{corollary-char-gi-d-2}
Let $I$ be some instance for {\sc Max d-KP}  which has an equivalent graph $G$.
Then $G$ is isomorphic to $G_d(I)$.
\end{corollary}

%%%%%%%%%%%%%%%%%%%%%%%%%%%%%%%%%%%%%%%%%%%%%%%%%%%%%%%%%%%%%%%%%%%%%%%%%%%%%%%%%%%%%%%%%%
\subsection{Counting and enumerating maximal independent sets}\label{sec-c-e-k-t}
%%%%%%%%%%%%%%%%%%%%%%%%%%%%%%%%%%%%%%%%%%%%%%%%%%%%%%%%%%%%%%%%%%%%%%%%%%%%%%%%%%%%%%%%%%

\subsubsection{$k$-threshold graphs}

Next we show how to enumerate and count the 
maximal independent sets in a $k$-threshold graph $G$ using these
sets for $k$ covering threshold graphs of $G$.

\begin{figure}[ht]
\hrule
\medskip
\begin{tabbing}
xxx \= xxx \= xxx \= xxx \= xxx\= xxx \= xxx\= xxx\= xxx \= xxx \kill
$\MIS(G)=\emptyset$; ${\cal I}=\emptyset$\\
for ($i = 1$; $i\leq k$; $i++$) \\
\>    compute $\MIS(G_i)$ \>\>\>\>\>\>\>\> $\vartriangleright$ see Figure \ref{fig:mis_thres}  \\
for each $(M_1,\ldots,M_k)\in  \MIS(G_1)\times \ldots \times \MIS(G_k)$ \\
\>   ${\cal I}= {\cal I} \cup \{M_1\cap \ldots \cap M_k\}$       \>\>\>\>\>\>\>\> $\vartriangleright$ compute ${\cal I}=\{M_1\cap \ldots \cap M_k~|~ M_i\in \MIS(G_i)\}$\\
for each $(I,J)\in {\cal I}\times {\cal I}$  \> \>\>\>\>\>\>\>\> $\vartriangleright$ remove all sets which are  subsets\\
\>  if ($I$ is a subset of $J$)   \\
\> \>  ${\cal I}= {\cal I}-I$\\
$\MIS(G)={\cal I}$;
\end{tabbing}
\hrule
\caption{Enumerating all maximal independent sets in a $k$-threshold graph.}
\label{fig:mis_thres-k}
\end{figure}

\begin{theorem}\label{t-enum-k2}
All maximal independent sets in a $k$-threshold graph $G$ on $n$ vertices
given by an edge set covering of
$k$ threshold graphs $G_i=(V,E_i)$ on $m_i$ edges, $1\leq i\leq k$,
can be enumerated and counted in time 
$\bigo(\sum_{i=1}^k m_i+n\cdot( \prod_{i=1}^k\omega(G_i))^2)\subseteq \bigo(n^{2k+1})$.
\end{theorem}

\begin{proof}
Let $G$ be a $k$-threshold graph on $n$ vertices and $m$ edges. 
Further let $G_i=(V,E_i)$, $1\leq i\leq k$, be
a covering by $k$ threshold graphs for $G$ and $m_i$ denote the number
of edges in $G_i$. By the method
given in Figure \ref{fig:mis_thres-k} we generate all maximal independent sets
in $G$.

The running time for computing $\MIS(G_i)$ for $1\leq i \leq k$ can
be bounded using Corollary \ref{cor-enum2} by 
$\bigo(\sum_{i=1}^k(\omega(G_i)\cdot n+m_i))=\bigo(n\cdot \sum_{i=1}^k\omega(G_i)+\sum_{i=1}^k m_i)$. 

Since every family $\MIS(G_i)$ consists of  $\omega(G_i)\leq n$ sets, there are
at most $\prod_{i=1}^k\omega(G_i) \leq n^k$ tuples $(M_1,\ldots,M_k)$ in $\MIS(G_1)\times \ldots \times \MIS(G_k)$.
For every tuple $(M_1,\ldots,M_k)$ the intersection $M_1 \cap \ldots \cap M_k$ can be computed
as follows. We have given $k$ subsets of a set (the vertex set of $G$) of $n$ elements.  
These $k$ subsets are merged into one set $M$ of at most $k\cdot n$ elements. We 
can sort $M$ using counting sort in time $\bigo(k\cdot n + n)=\bigo(k\cdot n)$.
Then an element of $M$ belongs to the intersection if and only if it occurs
$k$ times consecutively in the sorted list $M$. This can be checked by comparing
every first element on position $i$ with the element on position $i+k-1$.
Thus one intersection $M_1 \cap \ldots \cap M_k$  can be computed in time $\bigo(k\cdot n)$
and the set of all intersections ${\cal I}$ can be computed in 
time $\bigo(k\cdot n\cdot ( \prod_{i=1}^k\omega(G_i)))$.

Then we have to eliminate non-maximal subsets in ${\cal I}$, 
where $|{\cal I}|\leq \prod_{i=1}^k\omega(G_i) \leq n^k$. 
For every
pair $(I,J)\in {\cal I}\times {\cal I}$ we can sort $I$ and $J$ and then
verify whether $I$ is a subset of $J$ in time $\bigo(n)$. Thus we 
can obtain $\MIS(G)$ from ${\cal I}$ in time 
$\bigo(n\cdot( \prod_{i=1}^k\omega(G_i))^2)$.

By assuming  $m_i\geq 1$ and thus $\omega(G_i)\geq 2$ for $1\leq i \leq k$ the overall running time is in 
$$\begin{array}{lcl}
& &\bigo(n\cdot \sum_{i=1}^k\omega(G_i)+\sum_{i=1}^k m_i+k\cdot n\cdot ( \prod_{i=1}^k\omega(G_i))+n\cdot( \prod_{i=1}^k\omega(G_i))^2)\\
&\subseteq& \bigo(n\cdot \prod_{i=1}^k\omega(G_i)+\sum_{i=1}^k m_i+ ( \prod_{i=1}^k\omega(G_i))\cdot n\cdot ( \prod_{i=1}^k\omega(G_i))+n\cdot( \prod_{i=1}^k\omega(G_i))^2) \\
&\subseteq& \bigo( \sum_{i=1}^k m_i+n\cdot( \prod_{i=1}^k\omega(G_i))^2)  \\
&\subseteq&\bigo(n^{2k+1}).
\end{array}$$ 

The correctness holds as follows. Every independent set $S$ in $G$ is also
an independent set in graph $G_i$ for every $1\leq i \leq k$ by the definition of $G$. 
Thus every independent set $S$ in $G$ is a subset
of some {\em maximal} independent set $M_i$ in graph $G_i$ for every $1\leq i \leq k$.
Thus every independent set $S$ in $G$ is a subset
of the intersection $M_1 \cap \ldots \cap M_k$ for some maximal independent sets 
$M_i$ in graph $G_i$ for every $1\leq i \leq k$.
Further since every such intersection $M_1 \cap \ldots \cap M_k$  is an independent set in $G$
and we remove the non-maximal independent sets from the set of all these intersections
in the last step of our method, we create exactly the set of all  maximal independent sets
of $G$. 
\end{proof}

\begin{figure}[ht]
\centerline{\epsfig{figure=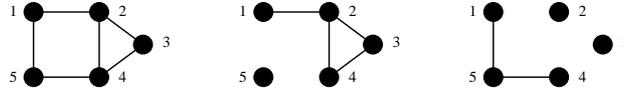,width=8.2cm}}
\caption{Covering the edge set of the house graph by two threshold graphs $\text{paw}\cup K_1$
and $P_3\cup 2K_1$ which is used in Example \ref{mis-2th-0}.}
\label{fig:house-cover}
\end{figure}

\begin{example}\label{mis-2th-0}   
We apply the algorithm  given in Figure \ref{fig:mis_thres-k}
to enumerate maximal independent sets in the 2-threshold graph \gansfuss{house}.
By Figure \ref{fig:house-cover} the edge set of the house can be covered
by two threshold graphs $G_1$ and $G_2$.
\begin{enumerate}
\item The algorithm shown in Figure \ref{fig:mis_thres} leads 
$\MIS(G_1)=\{\{2,5\},\{1,4,5\},\{1,3,5\}\}$ and  $\MIS(G_2)=\{\{2,3,5\},\{1,2,3,4\}\}$.
\item  ${\cal I}=\{\{2,5\},\{2\},\{5\},\{1,4\},\{3,5\},\{1,3\}\}$
\item $\MIS(\text{house})=\{\{2,5\},\{1,4\},\{3,5\},\{1,3\}\}$
\end{enumerate}
\end{example}

Thus it holds $\IM(\text{house})=\MIS(\text{house})$. In order to give an example for
a 2-threshold graph where thus equality does not hold we consider the gem graph
which is a split graph as well as a 2-threshold graph.
The problem of computing the set of all maximum
independent sets is considered in Section 
\ref{sec-maximum-is-kt}.

\begin{figure}[ht]
\centerline{\epsfig{figure=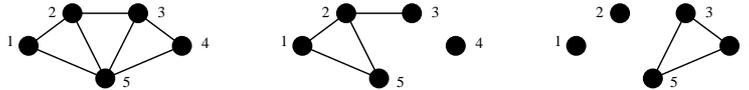,width=9.9cm}}
\caption{Covering the edge set of the gem graph by two threshold graphs $\text{paw}\cup K_1$
and $K_3\cup 2K_1$.}
\label{fig:gem-cover}
\end{figure}

\begin{example}  
We apply the algorithm  given in Figure \ref{fig:mis_thres-k}
to enumerate maximal independent sets in the 2-threshold graph \gansfuss{gem}.
By Figure \ref{fig:gem-cover} the edge set of the house can be covered
by two threshold graphs $G_1$ and $G_2$.
\begin{enumerate}
\item The algorithm shown in Figure \ref{fig:mis_thres} leads 
$\MIS(G_1)=\{\{3,4,5\},\{1,3,4\},\{2,4\}\}$ and  
$\MIS(G_2)=\{\{1,2,3\},\{1,2,4\},\{1,2,5\}\}$.
\item  ${\cal I}=\{\{3\},\{4\},\{5\},\{1,3\},\{1,4\},\{1\},\{2\},\{2,4\}\}$
\item $\MIS(\text{gem})=\{\{5\},\{1,3\},\{1,4\},\{2,4\}\}$
\end{enumerate}
\end{example}

The method given in Figure \ref{fig:mis_thres-k} implies the following bound.

\begin{corollary}\label{k-th-count} 
Let $G$ be a $k$-threshold graph and $G_i=(V,E_i)$, $1\leq i\leq k$, a 
covering by $k$ threshold graphs for $G$.
Then $G$ has at most $\prod_{i=1}^k \omega(G_i)$ maximal independent sets.
\end{corollary}

In order to show that the bound $\mis(G)\leq \prod_{i=1}^k \omega(G_i)$ can be achieved, 
we consider $G$ as the disjoint union of $k$ complete graphs $K_\ell$. 
Then for $\ell\geq 2$ the graph $G$ has has threshold dimension $k$. 
Further $G$ leads $\ell^k=\prod_{i=1}^k \omega(K_\ell)$ maximal independent sets.

For several graphs we can bound the size of the largest 
complete subgraph.

\begin{corollary}\label{k-th-count-planar-and-so} Let $G$ be a $k$-threshold graph
and $G_i=(V,E_i)$, $1\leq i\leq k$, 
some $k$ covering threshold graphs for $G$.
\begin{enumerate}
\item If every $G_i$, $1\leq i \leq k$, is planar, then $\mis(G)\leq 4^k$. 
\item If every $G_i$, $1\leq i \leq k$, is uniformly $\ell$-sparse, then $\mis(G)\leq(2\ell+1)^k$. 
\item If every $G_i$, $1\leq i \leq k$, has maximum degree at most $d$, then $\mis(G)\leq(d+1)^k$. 
\end{enumerate}
\end{corollary}

\begin{proof}\begin{inparaenum}[(1.)]
\item Planar graphs do not contain the $K_5$ as a subgraph, thus
$\omega(G_i)\leq 4$ for  $1\leq i\leq k$.
\item If a graph $G_i$ is uniformly $\ell$-sparse then the complete 
graph $K_{2\ell+2}$ is not a subgraph of $G_i$, thus $\omega(G_i)\leq 2\ell+1$ for  $1\leq i\leq k$.
\item If a graph $G_i$ has maximum degree at most $d$ then the complete 
graph $K_{d+2}$ is not a subgraph of $G_i$, thus $\omega(G_i)\leq d+1$ for  $1\leq i\leq k$.
\end{inparaenum}
\end{proof}

If some $k$-threshold graph has one of the discussed properties,
then this also holds for all of the $k$ subgraphs $G_i$.

\begin{corollary}\label{k-th-count-planar-and-so-on} Let $G$ be a $k$-threshold graph.
\begin{enumerate}
\item If $G$ is planar, then $\mis(G)\leq 4^k$.
\item If $G$ is uniformly $\ell$-sparse, then $\mis(G)\leq(2\ell+1)^k$.
\item If $G$ has maximum degree at most $d$, then $\mis(G)\leq (d+1)^k$.
\end{enumerate}
\end{corollary}

%%%%%%%%%%%%%%%%%%%%%%%%%%%%%%%%%%%%%%%%%%% Enumerate 2-threshold

The method given in Figure \ref{fig:mis_thres-k} leads all maximal independent sets
for $k$-threshold graphs and  an interesting bound on the
number of maximal independent sets using the clique number of 
$k$ suitable covering threshold graphs (Corollary \ref{k-th-count}).
The main drawback of the method
given in Figure \ref{fig:mis_thres-k} is the last step in which we have to remove  
non-maximal subsets.
For better solutions on enumerating all maximal independent sets in $k$-threshold graphs
we consider the case $k=2$. The main idea is to partition the vertex set of a 
2-threshold graph into three cliques and one independent set. This allows
us to generalize the idea for split graphs given in the proof of Theorem  \ref{l-is-spli}.

\begin{table}
\[
\renewcommand{\arraystretch}{1.8}
\begin{array}{l|c|c|}
    &  ~ K_2  ~ & ~ S_2 ~   \\
\hline
~K_1~ &    K    &   A    \\
\hline
~S_1~ &     B    &   S   \\
\hline
\end{array}
\]
\caption{\label{abb-part}Partition of the vertex set of a 2-threshold graph 
in the proof of Theorem \ref{t-enum-2}.}
\end{table}

\begin{theorem}\label{t-enum-2}
All maximal independent sets in a 2-threshold graph on $n$ vertices
can be enumerated in time $\bigo(n^3)$ and counted in time $\bigo(n^2)$.
\end{theorem}

\begin{proof}
Let $G=(V,E)$ be a 2-threshold graph. Then there are two threshold graphs 
$G_1=(V,E_1)$ and $G_2=(V,E_2)$ that cover the edge set of $G$, i.e. $E=E_1\cup E_2$.
Let $(K_1,S_1)$ be a split partition of $G_1$ and 
$(K_2,S_2)$ be a split partition of $G_2$. Such 
partitions are easy to find by Remark \ref{tresh-spli-p}. 
We consider
the intersections $K=K_1\cap K_2$, $S=S_1\cap S_2$, $A=K_1\cap S_2$, and $B=S_1\cap K_2$.
Then $K$, $A$, and $B$ are cliques in $G$ and 
$S$ is an independent set in $G$. 
Further it holds
$V=K\cup S \cup A \cup B$, see Table \ref{abb-part}.
We define $K'$ ($A'$, $B'$, respectively) as the subset of the vertices in $K$ ($A$, $B$, respectively) which are not adjacent to any vertex in $S$.

The maximal independent sets of $G$ can be found as follows.
\begin{enumerate}
\item Set $S$ is an independent set. 
But $S$ has not be maximal (even not if $S_1$ and $S_2$ are maximal). That is,
we have to verify whether we
we can add vertices from $A'$, $B'$, or $K'$ to $S$. 
Since $K_1$  and $K_2$ are cliques in $G_1$ and $G_2$ 
they are also cliques in $G$ and thus $A\cup K$ is a clique in $G$ as well
$B\cup K$. Thus we cannot add a vertex from $A'$ and a vertex from $K'$ 
to an independent set and the same holds for $B'$ and $K'$ and for all
three sets $A'$, $B'$, and $K'$. We can either add a vertex from $K'$ or
at most two nonadjacent vertices from $A'\cup B'$.
%non adjacent impliziert ja dass einer aus a und einer aus b ist
\begin{enumerate}
\item 
If $K'=A'=B'=\emptyset$, then $S$ is a  maximal independent set.

\item 
If $K'$ is non-empty, for every  $v\in K'$  we obtain a maximal independent set $S\cup\{v\}$.

\item 
If either $A'$ or $B'$ is empty, the other one will be
treated as $K'$. 

\item 
If $A'$ and  $B'$ are non-empty, 
for every  $(v_1,v_2)\in A'\times B'$ such that $\{v_1,v_2\}\not \in E$
we obtain a maximal independent set $S\cup\{v_1,v_2\}$.

\item 
If $A'$ and  $B'$ are non-empty, for every  $v_1 \in A'$ such that $v_1$ is adjacent
to every vertex in $B'$, 
we obtain a maximal independent set $S\cup\{v_1\}$.

\item 
If $A'$ and  $B'$ are non-empty, for every  $v_2 \in B'$ such that $v_2$ is adjacent
to every vertex in $A'$, 
we obtain a maximal independent set $S\cup\{v_2\}$.
\end{enumerate}

In this cases we have to add up to two vertices into $S$
to obtain a maximal independent set.

\item Every proper subset of $S$ is an independent set but not
a maximal independent set. 
But the possibilities how we can add vertices from $V-S$ to subsets of $S$
are restricted as mentioned in case (1.). Obviously a vertex $v\in V-S$ 
can only be added to the  (possibly empty) set of its non-neighours in $S$, which we denote by
$\overline{N(v,S)}$.

\begin{enumerate}
\item  
If $K-K'$ is non-empty, for every  $v\in K-K'$  we obtain a maximal independent set $\overline{N(S,v)}\cup\{v\}$.

%every vertex $v\in K-K'$ together with its non-neighbours
%in $S$ leads a maximal independent set.

\item 
If either $A-A'$ or $B-B'$ is empty, the other one will be 
treated as $K-K'$. 

\item 
If $A-A'$ and $B-B'$ are non-empty, for every  
$(v_1,v_2)\in A-A'\times B-B' \cup A'\times B-B'\cup A-A'\times B'$ such that $\{v_1,v_2\}\not \in E$
we obtain a maximal independent set by 
$(\overline{N(v_1,S)}\cap \overline{N(v_2,S)})\cup\{v_1,v_2\}$.

\item 
For every  $v_1 \in A-A'$ such that there is no $v_2\in B-B'$ such that $\{v_1,v_2\}\not \in E$
and $N(v_2,S)\subseteq N(v_1,S)$ (in order to ensure maximality of the sets)
we obtain a maximal independent set  $\overline{N(v_1,S)}\cup\{v_1\}$.

\item 
For every  $v_2 \in B-B'$  such that there is no $v_1\in A-A'$ such that $\{v_1,v_2\}\not \in E$
and $N(v_1,S)\subseteq N(v_2,S)$ (in order to ensure maximality of the sets)
we obtain a maximal independent set $\overline{N(v_2,S)}\cup\{v_2\}$.

\end{enumerate}
In this cases we have to add up to two vertices into a subset of $S$
to obtain a maximal independent set.

\end{enumerate}
By our construction all independent sets obtained in both steps are maximal.
The union of all these leads the family
of all maximal independent sets of $G$. Every of the eleven conditions can be verified in time $\bigo(n^2)$
and leads $\bigo(n^2)$ maximal independent sets of size $\bigo(n)$.
\end{proof}

The partition of the vertex set of a 2-threshold graph in the proof of Theorem \ref{t-enum-2}
motivates to look at the following generalization of split graphs which
was introduced by Brandst\"adt in \cite{Bra96}. 
A graph is a {\em $(k,\ell)$-graph} if its vertex set can be partitioned into $k$ 
independent sets and $\ell$ cliques.
For the edges between vertices of these sets there is no restriction.
Thus $(2,0)$-graphs are exactly bipartite graphs, $(k,0)$ are exactly $k$-colorable 
graphs, and $(1,1)$-graphs are exactly split graphs.
The given partition  the proof of Theorem \ref{t-enum-2} 
shows that 2-threshold graphs are special $(1,3)$-graphs and
more generally  $k$-threshold graphs are special $(1,2^{k-1})$-graphs.

\begin{proposition}\label{prop-classes2} 
We have the following properties for $k$-threshold graphs and for $k$-threshold intersection graphs.
\begin{enumerate}
\item $k$-threshold $\subset$  $(1,2^{k-1})$-graphs
\item $k$-threshold intersection $\subset$  $(2^{k-1},1)$-graphs
\end{enumerate}
\end{proposition}

%%%%%%%%%%%%%%%%%%%%%%%%%%%%%%%%%%%%%%%%%%%%%%%%%%%%%%%%%%%%%%%%%%%%%%
\subsubsection{Solutions for the Multidimensional Knapsack Problem}
%%%%%%%%%%%%%%%%%%%%%%%%%%%%%%%%%%%%%%%%%%%%%%%%%%%%%%%%%%%%%%%%%%%%%%

We have shown how to solve every instance  for {\sc Max KP} on $n$ items 
which has an equivalent graph in
time $\bigo(n^2)$ (cf.\ Theorem \ref{main-th-kp}) using
the method given in Figure \ref{fig:mis_thres}.
Next we want to give a similar result for  instances  for {\sc Max d-KP}
using our method  for
enumerating maximal independents in a $d$-threshold
graph shown 
in Figure \ref{fig:mis_thres-k}. 

%uses the 
%maximal independents within a set of $d$ covering
%threshold graphs, we have to ensure that all 
%instances $I_1,\ldots,I_d$ have an equivalent graph. 
%Example \ref{ex-no-kp-i3} shows that this can not be
%ensured by the existence of an equivalent graph
%of the instance for {\sc Max d-KP}.

\begin{theorem}\label{solve-dkp}
Let $I$ be an instance for {\sc Max d-KP} on $n$ items which has an equivalent graph.
Then $I$ can be solved in time $\bigo(n^{2d+1})$.
\end{theorem}

\begin{proof}
Let  $I$ be some instance for {\sc Max d-KP}  on $n$ items such that $I$ has an 
equivalent graph. By Theorem \ref{corollary-char-d} instance $I$ is equivalent
to graph $G_d(I)$. Graph $G_d(I)$ is a $d$-threshold graph by Observation \ref{obs-g-i-d}.
In order to apply Theorem \ref{t-enum-k2} on $G_d(I)$ we need 
$d$ covering threshold graphs $G_i$, $1\leq i\leq d$, for $G_d(I)$.
Such graphs can be obtained by $G_i=G(I_i)$, i.e. the
graphs defined in Section \ref{sec-graph-fr-pr} for every of the 
$d$ instances $I_i$ for  {\sc Max KP} defined in Section \ref{sec-k-t-f-mkp}.
Every graph $G_i$ is a threshold graph by Observation \ref{obs-g-i} and
can be constructed in time $\bigo(n^2)$. 

Thus the $\prod_{i=1}^d \omega(G(I_i))\leq n^d$ maximal independent sets in $G$ 
can be found in time $\bigo(n^{2d+1})$ by Theorem \ref{t-enum-k2} and correspond to the 
 maximal feasible solutions of $I$. For every of these solutions
we can compute its profit in time $\bigo(n)$.
\end{proof}

A related approach to solve special for {\sc Max d-KP} 
instances was suggested by Chv{\'a}tal and Hammer in \cite{CH77}.
They consider $m\times n$  zero-one matrices  $A$ for which
there is a single inequality
\begin{equation}
\sum_{j=1}^n a_jx_j\leq b
\end{equation}
whose  zero-one solutions are exactly the zero-one
solutions of the $m$ inequalities $i=1,\ldots,m$
\begin{equation}
\sum_{j=1}^n a_{i,j}x_j\leq 1.
\end{equation}
For such a matrix $A$ they define a graph $G(A)$ by 
representing the columns of $A$ as vertices 
and two vertices are adjacent if and only if 
the dot product of the corresponding vectors is positive.
If $G(A)$ is a threshold graph the {\sc Max d-KP} instance $I_A$
using zero-one sizes with respect to $A$, capacities $c_j=1$ and
$d=m$ dimensions
was  solved in \cite{CH77} within time $\bigo(m\cdot n^2)$ by using the split 
partition of graph $G(A)$.

The sketched results of \cite{CH77} motivate us to 
consider {\sc Max d-KP} instances $I$
which have an equivalent threshold graph $G$. 
By Corollary \ref{corollary-char-gi-d-2} we
know that $G$ is isomorphic to $G_d(I)$. 
Since $G$ is a threshold graph in this case 
$G_d(I)$ is also a threshold graph. Thus
instead of the method shown 
in Figure \ref{fig:mis_thres-k} we now can apply
the method given in Figure \ref{fig:mis_thres} on 
graph $G_d(I)$ in order to list all maximal feasible solutions
for $I$. Since $G_d(I)$ can be defined from $I$ in
time  $\bigo(d\cdot n^2)$ we have shown the following result.

\begin{corollary}
Let $I$ be an instance for {\sc Max d-KP} on $n$ items such that $I$ 
has an equivalent threshold graph.
Then $I$ can be solved in time $\bigo(d\cdot n^2)$.
\end{corollary}

Comparing the running $\bigo(d\cdot n^2)$ time with the result in  \cite{CH77}
we obtain the same running time but we even can handle instances using positive 
integer valued sizes and capacities.
Further $\bigo(d\cdot n^2)$ is no longer exponential within $d$ as in
Theorem \ref{solve-dkp} but only solves much more restricted {\sc Max d-KP}
instances.

Theorem \ref{t-enum-2} allows to improve Theorem \ref{solve-dkp}
for $d=2$ dimensions:

%Theorem \ref{t-enum-2} allows to get a better result as in Theorem \ref{solve-dkp}
%for $d=2$ dimensions:

\begin{theorem}
Let $I$ be an instance for {\sc Max 2-KP} on $n$ items which has 
an equivalent graph.
Then $I$ can be solved in time $\bigo(n^{3})$.
\end{theorem}

%%%%%%%%%%%%%%%%%%%%%%%%%%%%%%%%%%%%%%%%%%%%%%%%%%%%%%%%%%%%%%%%%%%%%%%%%%%
\subsection{Counting and enumerating maximum independent sets}\label{sec-maximum-is-kt}
%%%%%%%%%%%%%%%%%%%%%%%%%%%%%%%%%%%%%%%%%%%%%%%%%%%%%%%%%%%%%%%%%%%%%%%%%%%

Since every maximum independent set is a maximal independent set,
our results given in Section \ref{sec-c-e-k-t} also can be applied to list all maximum independent sets
in $k$-threshold graphs. 
By omitting the last step of 
the method given in Figure \ref{fig:mis_thres-k} and removing 
non-maximum sets we obtain a
method of running time $\bigo(\sum_{i=1}^k m_i+k\cdot n\cdot ( \prod_{i=1}^k\omega(G_i)))\subseteq \bigo(k\cdot n^{k+1})$ for listing 
all maximum independent sets
in a $k$-threshold graph.

\begin{corollary} \label{k-th-indep}
All maximum independent sets in a $k$-threshold graph $G$ on $n$ vertices
given by an edge set covering of
$k$ threshold graphs $G_i=(V,E_i)$ on $m_i$ edges, $1\leq i\leq k$,
can be enumerated and counted in time 
$\bigo(\sum_{i=1}^k m_i+k\cdot n\cdot ( \prod_{i=1}^k\omega(G_i)))\subseteq \bigo(k\cdot n^{k+1})$.
The size of a maximum independent set, i.e. $\alpha(G)$, in a $k$-threshold graph
$G$ can be computed in the same time.
\end{corollary}

The related problem of finding {\em one} 
independent set of maximum size in a  $k$-threshold graph
was solved in \cite{CLR04} in time $\bigo(n\log n + n^{k-1})$.
Comparing our solutions and these of \cite{CLR04} we observe that
we require graph representations and the authors of \cite{CLR04}
use the coefficients occurring in the multiple knapsack instance. Each of these
versions can transformed into the other
by Lemma \ref{le-gr-k2} and \ref{le-k-gr2}. Especially when
we can bound the vertex degree of the threshold graphs (cf.\ Corollary \ref{k-th-count-planar-and-so})
our results are much better.

%%%%%%%%%%%%%%%%%%%%%%%%%%%%%%%%%%%%%%%%%%%%%%%%%%%%%%%%%%%%%%%%%%%%%%%%%%%
\subsection{Counting and enumerating  independent sets}\label{sec-is-kt}
%%%%%%%%%%%%%%%%%%%%%%%%%%%%%%%%%%%%%%%%%%%%%%%%%%%%%%%%%%%%%%%%%%%%%%%%%%%

Next we show how to enumerate and count the 
maximal independent sets in a $k$-threshold graph $G$ using these
sets for $k$ covering threshold graphs of $G$.

\begin{figure}[ht]
\hrule
\medskip
\begin{tabbing}
xxx \= xxx \= xxx \= xxx \= xxx\= xxx \= xxx\= xxx\= xxx \= xxx \kill
$\IS(G)=\emptyset$;\\
for ($i = 1$; $i\leq k$; $i++$) \\
\>    compute $\IS(G_i)$ \>\>\>\>\>\>\>\> $\vartriangleright$ see Figure \ref{fig:is_thres} \\
for each $(M_1,\ldots,M_k)\in  \IS(G_1)\times \ldots \times \IS(G_k)$ \\
\>   $\IS(G)= \IS(G) \cup \{M_1\cap \ldots \cap M_k\}$       \>\>\>\>\>\>\>\> $\vartriangleright$ compute $\IS(G)=\{M_1\cap \ldots \cap M_k~|~ M_i\in \IS(G_i)\}$
\end{tabbing}
\hrule
\caption{Enumerating all independent sets in a $k$-threshold graph.}
\label{fig:is_thres-k}
\end{figure}

\begin{corollary} \label{k-th-inde}
All independent sets in a $k$-threshold graph $G$ on $n$ vertices
given by an edge set covering of
$k$ threshold graphs $G_i=(V,E_i)$ on $m_i$ edges, $1\leq i\leq k$,
can be enumerated and counted in time 
$\bigo(k\cdot n\cdot 2^{n\cdot k})$.
\end{corollary}

\begin{proof}
%anpassen m nicht ben
Let $G$ be a $k$-threshold graph on $n$ vertices.
Further let $G_i=(V,E_i)$, $1\leq i\leq k$, be
a covering by $k$ threshold graphs for $G$. By the method
given in Figure \ref{fig:is_thres-k} we generate all  independent sets in $G$. 
The running time for computing $\IS(G_i)$ for $1\leq i \leq k$ can
be bounded using Corollary \ref{t-i-enum2} by $\bigo(k\cdot n\cdot 2^{n-1})$. 
Since every family $\IS(G_i)$ consists of at most $2^{n}$ sets, there are
at most $2^{n\cdot k}$ tuples $(M_1,\ldots,M_k)$ in $\IS(G_1)\times \ldots \times \IS(G_k)$.
For every tuple $(M_1,\ldots,M_k)$ the intersection $M_1 \cap \ldots \cap M_k$ can be computed
in time $\bigo(k\cdot n)$ (cf.\ proof of Theorem \ref{t-enum-k2})
and the set of all intersections ${\cal I}$ can be computed in 
time $\bigo(k\cdot n\cdot 2^{n\cdot k})$.
Identical sets can be observed by suitable data structure.
The overall running time can be bounded by $\bigo(k\cdot n\cdot 2^{n\cdot k})$.
\end{proof}

%%%%%%%%%%%%%%%%%%%%%%%%%%%%%%%%%%%%%%%%%%%%%%%%%%%%%%%%%%%%%%%%%%%%%%%%%%%
\subsection{$d$-dimensional Vector Packing}\label{sec-vbp}
%%%%%%%%%%%%%%%%%%%%%%%%%%%%%%%%%%%%%%%%%%%%%%%%%%%%%%%%%%%%%%%%%%%%%%%%%%%

Next we consider the generalization of {\sc Min BP} where 
every item $a_j$ corresponds to a $d$-dimensional vector $s_j:=(s_{1,j},\ldots,s_{d,j})$
and every bin $i$ to the unit $d$-dimensional vector $(1,\ldots,1)$.
A set of items $A'$ {\em fits} into a bin if each component of the vector 
$\sum_{a_j\in A'}s_j$ does not exceed $1$.

\begin{desctight}
\item[Name] {\sc Min d-dimensional Vector Packing} ({\sc Min d-VP})

\item[Instance] A set $A=\{a_1,\ldots,a_n\}$ of $n$ items, a number $d$ of dimensions,
and a number $n$ of $d$-dimensional bins. 
Every item $a_j$ has $d$-dimensional positive rational size vector $s_j:=(s_{1,j},\ldots,s_{d,j})$. 

\item[Task] Find $n$ disjoint (possibly empty) subsets $A_1,\ldots, A_n$ of $A$ such that 
the  number of non-empty subsets is minimized and the items of
each subset fit into a bin.
\end{desctight}

For some instance $I$ for {\sc Min d-VP}
two items $a_j$ and $a_{j'}$ can be chosen into the same subset (bin) if and only if 
\begin{equation}
s_{1,j}+s_{1,j'}\leq 1 \wedge s_{2,j}+s_{2,j'} \leq 1 \wedge 
\ldots  \wedge s_{d,j}+s_{d,j'}\leq 1.\label{eq-vbp}
\end{equation}
This also motivates to model the compatibleness of the items in 
$I$ by $d$-threshold graphs. Therefor we consider
for every dimension  $i$, $1\leq i \leq d$, the inequality 
\begin{equation}
s_{i,1}x_1+s_{i,2}x_2+ \ldots + s_{i,n}x_n \leq 1.\label{v-l1}
\end{equation}

An instance $I$ for {\sc Min d-VP} and a graph $G=(V,E)$
are {\em equivalent}, if there is a bijection $f:A\to V$ such that for every $A'\subseteq A$ 
the characteristic vector of $A'$ satisfies for every dimension  $i$, $1\leq i \leq d$, 
inequality (\ref{v-l1}) if and only if $f(A'):=\{f(a_j)~|~a_j\in A'\}$ is an independent 
set of $G$.

If for every dimension  $i$, $1\leq i \leq d$,  the {\sc Min BP}
instance defined by inequality (\ref{v-l1})
has an equivalent  graph $G_i=(V,E_i)$, see Section \ref{sec-bp}, we know that $G_i$
is a threshold graph. Further
two items $a_j$ and $a_{j'}$ 
can be chosen into the same bin if and only if
\begin{equation}
\{v_j,v_{j'}\}\not\in E_1 \wedge \{v_j,v_{j'}\}\not\in E_2   \wedge \ldots  \wedge \{v_j,v_{j'}\}\not\in E_d.\label{in-g2}
\end{equation}
Then $G=(V,E)$ where $E=E_1 \cup \ldots \cup E_d$ leads
a $d$-threshold graph. By (\ref{in-g2}) two
items $a_j$ and $a_{j'}$ can be chosen into the same bin if and only if
there is no edge $\{v_{j},v_{j'}\}$ in $G$. That is any two items
corresponding to the vertices of a clique in $G$ can not be chosen
into the same bin. Thus $\omega(G)$ leads a lower bound on the number of 
bins needed for the considered 
instance $I$ for {\sc Min d-VP}.

\begin{theorem}\label{bin-p1}
Let $I$ be an instance for {\sc Min d-VP} such that for every dimension $i$ instance $I_i$ 
has an equivalent graph $G_i=(V,E_i)$. Then graph $G=(V,E_1\cup\ldots \cup E_d)$ 
leads the bound $OPT(I)\geq \omega(G)$.
\end{theorem}

In order to know an equivalent graph for $I$ we can proceed as in Section \ref{sec-k-t-f-mkp}
and use graph $G_d(I)$ defined by the values of the $d$ inequalities (\ref{v-l1}).

\begin{theorem}\label{bin-p2}
Let $I$ be an instance for {\sc Min d-VP} such that for every dimension $i$ instance $I_i$ 
has an equivalent graph. Then $OPT(I)\geq \omega(G_d(I))$.
\end{theorem}

%%%%%%%%%%%%%%%%%%%%%%%%%%%%%%%%%%%%%%%%%%%%%%%%%%%%%%%%%%%%%%%%%%%%%%%%%%%
\subsection{$d$-dimensional Bin Packing}\label{sec-dbp}
%%%%%%%%%%%%%%%%%%%%%%%%%%%%%%%%%%%%%%%%%%%%%%%%%%%%%%%%%%%%%%%%%%%%%%%%%%%

Next we consider the generalization of {\sc Min BP}
where we have to pack items corresponding to
$d$-dimensional parallelepipads into bins corresponding to $d$-dimensio\-nal cubes 
of size $1$ on every dimension. 
Each of the items $a_j$ has for dimension $i$ the size $s_{i,j}$.
The items have to be {\em packed} without overlapping and without rotations
into the cubes, such that the faces of the items are parallel to those
of the cubes.

\begin{desctight}
\item[Name] {\sc Min d-dimensional Bin Packing} ({\sc Min d-BP})

\item[Instance] A set $A=\{a_1,\ldots,a_n\}$ of $n$ items, a number $d$ of dimensions,
and a number $n$ of $d$-dimensional bins. 
Every item $a_j$ has for dimension $i$ the positive rational size $s_{i,j}$. 
%Further there is a capacity $c$ for every dimension for each bin.

\item[Task] Find $n$ disjoint (possibly empty) subsets $A_1,\ldots, A_n$ of $A$ such that 
the number of non-empty subsets is minimized and the items of
each subset can be packed into a different $d$-dimensio\-nal cube.
\end{desctight}

For some instance $I$ for {\sc Min d-BP}
two items $a_j$ and $a_{j'}$ can
be chosen into the same subset (bin) if and only if
\begin{equation}
s_{1,j}+s_{1,j'}\leq 1 \vee s_{2,j}+s_{2,j'} \leq 1 \vee \ldots  \vee s_{d,j}+s_{d,j'}\leq 1.\label{d-l1x}
\end{equation}

This motivates to model the compatibleness of the items in 
instance $I$ by  $d$-threshold graphs. Therefor we consider
for every dimension  $i$, $1\leq i \leq d$, the inequality 
\begin{equation}
s_{i,1}x_1+s_{i,2}x_2+ \ldots + s_{i,n}x_n \leq 1.\label{d-l1}
\end{equation}

An instance $I$ for {\sc Min d-BP} and a graph $G=(V,E)$
are {\em equivalent}, if there is a bijection $f:A\to V$ such that for every $A'\subseteq A$ 
the characteristic vector of $A'$
satisfies for every dimension  $i$, $1\leq i \leq d$, inequality  (\ref{d-l1})
if and only if $f(A'):=\{f(a_j)~|~a_j\in A'\}$ is an independent 
set of $G$.

If for every dimension  $i$, $1\leq i \leq d$,  the  {\sc Min BP}
instance defined by inequality (\ref{d-l1})
has an equivalent graph $G_i=(V,E_i)$, see Section \ref{sec-bp}, then we know that $G_i$ is a threshold graph.
Further two items $a_j$ and $a_{j'}$ can be chosen into the same bin if and only if
\begin{equation}
\{v_j,v_{j'}\}\not\in E_1 \vee \{v_j,v_{j'}\}\not\in E_2   \vee \ldots  \vee \{v_j,v_{j'}\}\not\in E_d.\label{in-g}
\end{equation}
Then  $G=(V,E)$ where $E=E_1 \cap \ldots \cap E_d$ leads
a $d$-threshold intersection graph. By (\ref{in-g}) two
items $a_j$ and $a_{j'}$ can be chosen into the same bin if and only if
there is no edge $\{v_{j},v_{j'}\}$ in $G$. That is any two items
corresponding to the vertices of a clique in $G$ can not be chosen
into the same bin. Thus $\omega(G)$ leads a lower bound on the number of 
bins needed for the considered 
instance $I$ for {\sc Min d-BP}. The value $\omega(G)$ can  be
computed in polynomial time from $G$, see Section \ref{sec-concl}.

\begin{theorem}\label{bin-p1x}
Let $I$ be an instance for {\sc Min d-BP} such that for every dimension $i$ instance $I_i$ 
has an equivalent graph $G_i=(V,E_i)$. Then graph $G=(V,E_1\cap\ldots \cap E_d)$ 
leads the bound $OPT(I)\geq \omega(G)$.
\end{theorem}

In order to know a graph we can proceed similar to Section \ref{sec-k-t-f-mkp}
and use graph  $G'_d(I)=(V(I),E(I))$ defined by 
\begin{equation}
V(I)=\{v_j~|~a_j\in A\} ~~~ \text{ and } ~~~  E(I)=\{\{v_{j},v_{j'}\}~|~ s_{1,j}+s_{1,j'}> c_1 \wedge \ldots \wedge s_{d,j}+s_{d,j'}> c_d \}
\end{equation}
using the values of the $d$ inequalities (\ref{d-l1}).

\begin{theorem}\label{bin-p2x}
Let $I$ be an instance for {\sc Min d-BP} such that for every dimension $i$ instance $I_i$ 
has an equivalent graph. Then $OPT(I)\geq \omega(G'_d(I))$.
\end{theorem}

%%%%%%%%%%%%%%%%%%%%%%%%%%%%%%%%%%%%%%%%%%%%%%%%%%%%%%%%%%%%%%%%%%%%%%%%%%
%%%%%%%%%%%%%%%%%%%%%%%%%%%%%%%%%%%%%%%%%%%%%%%%%%%%%%%%%%%%%%%%%%%%%%%%%%
\section{Conclusions}\label{sec-concl}
%%%%%%%%%%%%%%%%%%%%%%%%%%%%%%%%%%%%%%%%%%%%%%%%%%%%%%%%%%%%%%%%%%%%%%%%%%
%%%%%%%%%%%%%%%%%%%%%%%%%%%%%%%%%%%%%%%%%%%%%%%%%%%%%%%%%%%%%%%%%%%%%%%%%%

We introduced methods to count and
enumerate all maximal independent, all maximum independent sets, and all independent 
sets in threshold graphs and $k$-threshold graphs.
These results allowed us to solve a large number of knapsack instances
in polynomial time. Since we generate all maximal independent sets for our
solutions we even can extend our results of Theorems \ref{main-th-kp} and \ref{solve-dkp} to generate
{\em all} optimal solutions of the respective knapsack instances.

The related problem of counting and listing maximal cliques in a threshold graph $G$
can be treated by considering maximal independent sets of the complement graph $\overline{G}$
by Lemma \ref{comple} and $\MC(G)=\MIS(\overline{G})$. 
Using Observation \ref{obs-creation-se} (\ref{obs-c-i})
the method shown in Figure \ref{fig:mis_thres} can be modified
to enumerate all maximal independent sets in a threshold graph. 
We have to exchange the
commands within the cases $t_i=1$ and $t_i=0$ and further we have
to store a maximum clique for $i=1$, since $t_1=1$ by our definition.
This implies that the number of all  maximal cliques in a threshold $G$ equals $\alpha(G)$.
If $G$ is a threshold graph $G$ on $n$ vertices 
which is given by a creation sequence, then all  maximal cliques
can be counted in time $\bigo(n)$ and enumerated  in time $\bigo(\alpha(G)\cdot n)$ and in
time $\bigo(1)$ per output.

In Table \ref{summary} we survey our results on
counting and enumerating special independent sets and cliques
within threshold graphs and $k$-threshold graphs. Since threshold graphs
are chordal (cf.\ Proposition \ref{prop-classes}) some
results are known from  \cite{OUU08}.

\begin{table}[ht]
\begin{center}
\scriptsize
\begin{tabular}{|l|ll|ll|}  
\hline
             & \multicolumn{2}{c|}{threshold graph $G$}      &  \multicolumn{2}{c|}{$k$-threshold graph $G$} \\
\hline
$\alpha(G)$ &  $\bigo(n+m)$  &  Corollary \ref{find-maxcln}  &  $\bigo(\sum_{i=1}^k m_i+k\cdot n\cdot ( \prod_{i=1}^k\omega(G_i)))\subseteq\bigo(k\cdot n^{k+1})$     &  Corollary \ref{k-th-indep}         \\
$\omega(G)$ &  $\bigo(n+m)$  &  by $\alpha(\overline{G})$    &  open                        &         \\
\hline
\hline
$\mis(G)$   &  $\bigo(n+m)$ &  Corollary \ref{cor-enum2}     &  $\bigo(\sum_{i=1}^k m_i+n\cdot( \prod_{i=1}^k\omega(G_i))^2)\subseteq\bigo(n^{2k+1})$           &  Theorem \ref{t-enum-k2}        \\
$\im(G)$    &  $\bigo(n+m)$ &  Corollary \ref{cor-enum2y} or \cite{OUU08}                  &  $\bigo(\sum_{i=1}^k m_i+k\cdot n\cdot ( \prod_{i=1}^k\omega(G_i)))\subseteq\bigo(k\cdot n^{k+1})$  &   Corollary \ref{k-th-indep}         \\
$\is(G)$    &  $\bigo(n+m)$ &  Corollary \ref{cor-enum-is} or \cite{OUU08}                  &  $\bigo(k\cdot n\cdot 2^{n\cdot k})$                            &   Corollary \ref{k-th-inde}         \\
\hline
$\mc(G)$    &    $\bigo(n+m)$ &  by $\mis(\overline{G})$    &  open                        &            \\
$\cm(G)$    &    $\bigo(n+m)$ &  by $\im(\overline{G})$     &  open                        &            \\
$\ic(G)$    &    $\bigo(n+m)$ &  by $\is(\overline{G})$     &  open                        &            \\
\hline
\hline
$\MIS(G)$   &  $\bigo(1)$     & Corollary \ref{cor-enum2}    &  $\bigo(\sum_{i=1}^k m_i+n\cdot( \prod_{i=1}^k\omega(G_i))^2)\subseteq\bigo(n^{2k+1})$  & Theorem \ref{t-enum-k2}\\
$\IM(G)$    &  $\bigo(1)$     & Corollary \ref{cor-enum2y} or \cite{OUU08}                 &$\bigo(\sum_{i=1}^k m_i+k\cdot n\cdot ( \prod_{i=1}^k\omega(G_i)))\subseteq\bigo(k\cdot n^{k+1})$                    & Corollary \ref{k-th-indep}\\
$\IS(G)$    &  $\bigo(1)$     &  Corollary \ref{cor-enum-is} or \cite{OUU08}                 &  $\bigo(k\cdot n\cdot 2^{n\cdot k})$         &  Corollary \ref{k-th-inde}\\
\hline
$\MC(G)$    &  $\bigo(1)$     &  by $\MIS(\overline{G})$     &  open                       & \\
$\CM(G)$    &  $\bigo(1)$     &  by  $\IM(\overline{G})$     &  open                       & \\
$\C(G)$     &  $\bigo(1)$     &  by  $\IS(\overline{G})$     &  open                       & \\
\hline
\end{tabular}
\end{center}
\caption
{Summary of results for counting and enumerating problems related to 
$k$-threshold graphs on $n$ vertices and $m$ edges. Running
times for enumeration algorithms for threshold graphs are given in time per output.}
\label{summary}
\end{table}

The structure of $k$-threshold graphs has shown to be very useful
for listing maximal independent sets. In a very similar way the structure
of $k$-threshold intersection graphs, i.e. graphs whose threshold intersection 
dimension is at most $k$, can be used to list maximal cliques. Therefore
we just have to compute $\MC(G_i)$  instead
of $\MIS(G_i)$ of the involved threshold graphs $G_i$ 
and combine them using the method 
shown in Figure \ref{fig:mis_thres-k}.
Thus all maximal cliques in a $k$-threshold intersection graph $G$ on 
$n$ vertices whose edge set is the intersection of those of $k$ threshold 
graphs $G_i=(V,E_i)$ on $m_i$ edges, $1\leq i\leq k$, can be enumerated and counted in time 
$\bigo(\sum_{i=1}^k m_i+n\cdot( \prod_{i=1}^k\alpha(G_i))^2)\subseteq \bigo(n^{2k+1})$.
Furthermore we know that $G$ has at most $\prod_{i=1}^k \alpha(G_i)$ maximal cliques.

The arguments given in Section \ref{sec-maximum-is-kt}
lead a method of running time $\bigo(\sum_{i=1}^k m_i+k\cdot n\cdot ( \prod_{i=1}^k\alpha(G_i)))\subseteq \bigo(k\cdot n^{k+1})$ 
for listing all maximum independent sets
in a $k$-threshold intersection graph. 
The size of a maximum clique, i.e. $\omega(G)$, in a $k$-threshold intersection graph
$G$ can be computed in the same time.

%%%%%%%%%%%%%%%%%%%%%%%%%%%%%%%%%%%%%%%%%%%%%%%%%%%%%%%%%%%%%%%%%%%%%%
\section{Acknowledgements} \label{sec-a}
%%%%%%%%%%%%%%%%%%%%%%%%%%%%%%%%%%%%%%%%%%%%%%%%%%%%%%%%%%%%%%%%%%%%%%

The work of the second author was supported by the German Research 
Association (DFG) grant  GU 970/7-1.

\bibliographystyle{alpha}





%\biboptions{longnamesfirst}

%\bibliographystyle{alpha}
%\bibliography{/home/gurski/bib.bib}

\begin{thebibliography}{AGRY16}

\bibitem[AGRY16]{AGRY16}
C.~Albrecht, F.~Gurski, J.~Rethmann, and E.~Yilmaz.
\newblock Knapsack {P}roblems: {A} {P}arameterized {P}oint of {V}iew.
\newblock {\em ACM Computing Research Repository (CoRR)}, abs/1611.07724:27
  pages, 2016.

\bibitem[BLS99]{BLS99}
A.~Brandst\"adt, V.B. Le, and J.P. Spinrad.
\newblock {\em Graph Classes: A Survey}.
\newblock SIAM Monographs on Discrete Mathematics and Applications. SIAM,
  Philadelphia, 1999.

\bibitem[Bra96]{Bra96}
A.~Brandst\"adt.
\newblock Partitions of graphs into one or two independent sets and cliques.
\newblock {\em Discrete Mathematics}, 152:47--54, 1996.

\bibitem[CH73]{CH73}
V.~Chv\'atal and P.L. Hammer.
\newblock Set-packing and threshold graphs.
\newblock Technical Report CORR 73-21, Comp. Sci. Dept., Univ. of Waterloo,
  1973.

\bibitem[CH77]{CH77}
V.~Chv\'atal and P.L. Hammer.
\newblock Aggregation of inequalities in integer programming.
\newblock {\em Annals of Discrete Math.}, 1:145--162, 1977.

\bibitem[CLR04]{CLR04}
A.~Caprara, A.~Lodi, and R.~Rizzi.
\newblock On $d$-threshold graphs and $d$-dimensional bin packing.
\newblock {\em Networks}, 44(4):266--280, 2004.

\bibitem[Fr{\'e}04]{Fre04}
A.~Fr{\'e}ville.
\newblock The multidimensional 0-1 knapsack problem: {A}n overview.
\newblock {\em European Journal of Operational Research}, 155:1--21, 2004.

\bibitem[Gol80]{Gol80}
M.C. Golumbic.
\newblock {\em Algorithmic Graph Theory and Perfect Graphs}.
\newblock Academic Press, 1980.

\bibitem[GR17]{GR17a}
F.~Gurski and C.~Rehs.
\newblock A graph theoretic approach to solve special knapsack problems in
  polynomial time ({A}bstract).
\newblock International Conference on Operations Research (OR 2017), 2017.

\bibitem[HIS78]{HIS78}
P.L. Hammer, T.~Ibaraki, and B.~Simeone.
\newblock Degree sequences of threshold graphs.
\newblock {\em Congressus Numerantium}, 21:329--355, 1978.

\bibitem[HK07]{HK07}
P.~Heggernes and D.~Kratsch.
\newblock Linear-time certifying recognition algorithms and forbidden induced
  subgraphs.
\newblock {\em Nord. J. Comput.}, 14(1-2):87--108, 2007.

\bibitem[HS81]{HS81}
P.L. Hammer and B.~Simeone.
\newblock The splittance of a graph.
\newblock {\em Combinatorica}, 1(3):275--284, 1981.

\bibitem[HSS06]{HSS06}
A.~Hagberg, P.J. Swart, and D.A. Schult.
\newblock Designing threshold networks with given structural and dynamical
  properties.
\newblock {\em Phys. Rev. E}, 056116, 2006.

\bibitem[KPP10]{KPP10}
H.~Kellerer, U.~Pferschy, and D.~Pisinger.
\newblock {\em Knapsack Problems}.
\newblock Springer-Verlag, Berlin, 2010.

\bibitem[Leu84]{Leu84}
J.~Y.-T. Leung.
\newblock Fast algorithms for generating all maximal independent sets of
  interval, circular-arc and chordal graphs.
\newblock {\em Journal of Algorithms}, 5(1):22--35, 1984.

\bibitem[MM65]{MM65}
J.~Moon and L.~Moser.
\newblock On cliques in graphs.
\newblock {\em Israel Journal of Mathematics}, 3:23--28, 1965.

\bibitem[MP95]{MP95a}
N.V.R. Mahadev and U.N. Peled.
\newblock {\em Threshold Graphs and Related Topics}.
\newblock Annals of Discrete Math. 56. Elsevier, North-Holland, 1995.

\bibitem[Orl77]{Orl77}
J.~Orlin.
\newblock The minimal integral separator of a threshold graph.
\newblock In B.H.~Korte P.L.~Hammer, E.L.~Johnson and G.L. Nemhauser, editors,
  {\em Studies in Integer Programming}, volume~1 of {\em Annals of Discrete
  Mathematics}, pages 415 -- 419. Elsevier, 1977.

\bibitem[OUU08]{OUU08}
Y.~Okamoto, T.~Uno, and R.~Uehara.
\newblock Counting the number of independent sets in chordal graphs.
\newblock {\em Journal of Discrete Algorithms}, 6(2):229--242, 2008.

\bibitem[PS09]{PS09}
U.~Pferschy and J.~Schauer.
\newblock The knapsack problem with conflict graphs.
\newblock {\em Journal of Graph Algorithms and Applications}, 13(2):233--249,
  2009.

\bibitem[PU59]{PU59}
M.~Paull and S.~Unger.
\newblock Minimizing the number of states in incompletely specified state
  machines.
\newblock {\em IRE Transactions on Electronic Computers}, EC-8:356--367, 1959.

\bibitem[Riv85]{Riv85}
I.~Rival, editor.
\newblock {\em Graphs and Order: {T}he Role of Graphs in the Theory of Ordered
  Sets and Its Applications}.
\newblock Nato Science Series C: (Book 147). Springer, 1985.

\bibitem[Rob97]{Rob97}
T.~Robinson.
\newblock Knapsack graphs.
\newblock {\em New Zealand Journal of Mathematics}, 26:107--123, 1997.

\bibitem[RRu89]{RRS89}
J.~Reitermann, V.~R{\"o}dl, and E.~\u{S}i\u{n}ajov\'{a}.
\newblock Geometrical embeddings of graphs.
\newblock {\em Discrete Mathematics}, 74:291--319, 1989.

\bibitem[SR98]{SR98}
A.~Sterbini and T.~Raschle.
\newblock An ${O}(n^3)$ time algorithm for recognizing threshold dimension 2
  graphs.
\newblock {\em Information Processing Letters}, 67(5):255--259, 1998.

\bibitem[Vad01]{Vad01}
S.P. Vadhan.
\newblock The complexity of counting in sparse, regular, and planar graphs.
\newblock {\em SIAM Journal on Computing}, 31(2):398--427, 2001.

\bibitem[Val79]{Val79}
L.G. Valiant.
\newblock The complexity of enumeration and reliability problems.
\newblock {\em SIAM Journal on Computing}, 8(3):410--421, 1979.

\bibitem[Yan82]{Yan82}
M.~Yannakakis.
\newblock The complexity of the partial order dimension problem.
\newblock {\em SIAM J. Algebraic Methods}, 3(3):351--358, 1982.

\bibitem[ZVI04]{OV04}
C.~Ortiz Z. and M.~Villanueva-Ilufi.
\newblock Difficult problems in threshold graphs.
\newblock {\em Electronic Notes in Discrete Mathematics}, 18:187 -- 192, 2004.

\end{thebibliography}

\end{document}